\documentclass[10pt]{amsart}
\usepackage[english]{babel}

\usepackage[letterpaper,top=2cm,bottom=2cm,left=3cm,right=3cm,marginparwidth=1.75cm]{geometry}

\usepackage[title]{appendix}
\usepackage{url}
\usepackage{amssymb,amsthm,amsmath}
\usepackage[numbers,sort&compress]{natbib}
\usepackage{color}
\usepackage{graphicx}
\usepackage{tikz}
\usepackage[colorlinks,linkcolor=blue,anchorcolor=blue,citecolor=blue]{hyperref}
\newtheorem{Theorem}{Theorem}[section]
\newtheorem{Lemma}[Theorem]{Lemma}
\newtheorem{proposition}[Theorem]{Proposition}
\newtheorem{corollary}[Theorem]{Corollary}
\newtheorem{remark}[Theorem]{Remark}

\newtheorem{definition}[Theorem]{Definition}

\numberwithin{equation}{section}

\newcommand{\R}{{\mathbb R}}

\newcommand{\Z}{{\mathbb Z}}

\title{Arithmetic phase transitions for Mosaic Maryland model}

\author{Jiawei He}
\address{Chern Institute of Mathematics and LPMC, Nankai University, Tianjin 300071, China}

\email{hermit\_well@163.com}

\author {Xu Xia}
\address{Chern Institute of Mathematics and LPMC, Nankai University, Tianjin 300071, China} 
\email{xiaxu14@mails.ucas.ac.cn}

\begin{document}
\maketitle

\begin{abstract}
We give a precise description of spectral types of the Mosaic Maryland model   with any irrational frequency,   which provides a quasi-periodic  unbounded  model with non-monotone potential has arithmetic phase transition. 
\end{abstract}

\section{Introduction}
In this paper, we study the spectral property of almost-periodic unbounded Schr\"odinger operator
\begin{equation*}
\left(H u\right)_{n}=u_{n+1}+u_{n-1}+\lambda v_n u_{n},
\end{equation*}
where $v_n$ is a unbounded almost-periodic sequence.  An important  example is 
\begin{equation*}
\left(H_{\lambda \tan , \alpha, \theta} u\right)_{n}=u_{n+1}+u_{n-1}+\lambda \tan \pi (\theta+n\alpha) u_{n},
\end{equation*}
 where $\lambda \in \mathbb{R}$ is the coupling
constant, $\alpha \in \mathbb{R}\backslash \mathbb{Q} $ is the frequency, $\theta\notin \frac{1}{2}+\alpha \mathbb{Z}+\mathbb{Z}$ is the phase. In the following, we just denote 
$\Theta \triangleq \frac{1}{2}+\alpha \mathbb{Z}+\mathbb{Z}$, and we say ``all $\theta$",  if $\theta \notin \Theta$.

This model  was  first proposed by Grempel, Fishman, and Prange in 1982 \cite{ref6} as a model stemming from the study of quantum chaos,
and later dubbed the Maryland model by B. Simon \cite{simon85}.
In recent years,  Maryland model has received extensive research due to the rich backgrounds in quantum physics \cite{berry1984, fishman10, ganeshan2014, longhi2021}.
The richness of its spectral theory, abundance of unusual features, and amenability to analysis also make it a crucial component of general conjectures and counterexamples in mathematics.

It is worth noting that if the potential is unbounded, there is no absolutely continuous spectrum for all $\theta$  \cite{simon1989}. As a result, it is natural to expect localized eigenfunctions.
Grempel-Fishman-Prange \cite{ref6} obtained in an essentially rigorous way, a dense set of explicitly determined eigenvalues, corresponding to exponentially decaying eigenfunctions, for Diophantine frequencies.  Here $\alpha$ is Diophantine, if there exist $\gamma, \sigma>0$, such that
$$ \inf _{j \in \mathbb{Z}}|n \alpha-j| \geq{\frac{\gamma}{|n|^{\sigma}}} \quad \forall n \neq 0 .$$
Indeed for Diophantine frequencies $\alpha$,  Maryland model has Anderson localization: pure point spectrum with exponentially decaying eigenfunctions, for all $\theta$ \cite{figotin1984, simon85}. 
Motivated by the Maryland model, Anderson localization for Maryland-type model have recently sparked tremendous interest in spectral theory of Schr\"odinger operator.  In general, $V$ is called Maryland-type potential, if $V$ is a  function 
\begin{equation*}
f:(-1 / 2,1 / 2) \rightarrow(-\infty,+\infty), \quad f(-1 / 2 \pm 0)=\mp \infty,
\end{equation*}
and can be extended into $\mathbb{R} \backslash(\mathbb{Z}+1 / 2)$ by 1-periodicity. We call  $V$  is  Lipschitz monotone if  there exists $\gamma>0$ such that $f(y)-f(x) \geqslant$ $\gamma(y-x)$ for all $0<x<y<1$. Using a KAM-type
procedure, Béllissard, Lima, and Scoppola   \cite{bellissard1983} obtained  Anderson localization for a class of meromorphic functions $V$ whose restrictions onto $\mathbb{R}$ are also 1-periodic and Lipschitz monotone with Diophantine frequencies. This  result \cite{bellissard1983} is perturbative, that is, once $V$ is fixed, one can only obtain localization for the potential $\lambda V$ with $\lambda \geqslant \lambda_{0}(\alpha)$, where $\lambda_{0}$ depends on the Diophantine constant of $\alpha$ and does not have a uniform lower bound on a full measure set of frequencies. Additionally, a large family of 1D quasiperiodic operators with unbounded monotone potentials and Diophantine frequencies were constructed in \cite{kachkovskiy2019}  as the non-perturbative outcome of Anderson localization. 
Another example of Maryland-type potential in a closely related setting is given in \cite{kachkovskiy2021}, where the authors show that Anderson localization can still be explored in operators with unbounded monotone potential, which are not required to be strictly monotone and can have flat segments under certain geometric conditions. As we can see, all  the preceding results all assume that the potential is monotonic, a natural question is whether or not a non-monotone $V$ exhibits Anderson localization. That is the primary motivation for our paper.

Our second motivation stems from the Maryland model's phase transition. Before we go any further, we consider another more famous quasi-periodic model:
\begin{equation*}
(H_{\lambda,\alpha,\theta} u)_n= u_{n+1}+u_{n-1} +2\lambda \cos 2
\pi (n\alpha + \theta) u_n,
\end{equation*}
This model is known as almost-Mathieu operator (or Aubry-Andre model in physical liteature), which is a bounded self-adjoint operator on $\ell^2(\Z)$. 
The Almost-Mathieu operator(AMO) was first  proposed by Peierls \cite{Pe},
as a model for an electron on a 2D lattice, acted on by a homogeneous
magnetic field \cite{Ha,R}.
AMO
undergoes a phase transition at $\lambda=1,$ where the Lyapunov
exponent changes from zero everywhere on the spectrum \cite{ref26} to
positive everywhere on the spectrum \cite{H}. Aubry-Andre conjectured
\cite{AA80}, that at $\lambda=1$
 the spectrum changes from absolutely continuous for $\lambda<1$
to pure point for $\lambda>1.$ This has since been proved, for all
$\alpha,\theta$ for $\lambda<1$ \cite{L93,Aab,AD,ref3}  and for Diophantine $\alpha,\theta$
(so a.e.) for $\lambda>1$ \cite{ref11}. Indeed, there exists second transition line from singular continuous spectrum to pure point spectrum. 
Let $p_{n} / q_{n}$ be the continued fraction approximates of $\alpha$. The index $\beta(\alpha)$ that measures Liouvilleness of the frequency is defined as follows:
\begin{equation*}
\beta(\alpha)=\limsup _{n \rightarrow \infty} \frac{\ln q_{n+1}}{q_{n}}.
\end{equation*}
If $ 1<\lambda<e^{\beta}$, then $H_{\lambda,\alpha,\theta}$ has purely singular continuous spectrum  for all $\theta$ \cite{ayz}, and if $\lambda>e^{\beta}$, then  $H_{\lambda,\alpha,\theta}$ has Anderson localization  for $\gamma(\alpha,\theta)=0$ \cite{ayz,ref84}, where we recall that
$$\gamma(\alpha,\theta)=-\limsup_{k\rightarrow \infty} \frac{\ln\|2\theta+k\alpha\|_{\R/\Z} }{|k|}.$$
Moreover, if $\ln|\lambda| <\gamma(\alpha,\theta),$ $H_{\lambda,\alpha,\theta}$ has purely singular continuous spectrum \cite{JLiu1}, and if $\ln|\lambda| >\gamma(\alpha,\theta),$ then $H_{\lambda,\alpha,\theta}$ has Anderson localization  for $\beta(\alpha)=0$ \cite{JLiu1}. To summarize, AMO has two different types of resonances: frequency resonances and phase resonances, where 
$\beta(\alpha)$ measures exponential strength of the frequency resonances, and $\gamma(\alpha,\theta)$ measure exponential strength of the phase resonances.  The second transition line claims that the operator displays localization when the Lyapunov exponent beats frequency/phase resonances.

Let us return to the Maryland model. As previously stated, if the frequency $\alpha$  is Diophantine, the Maryland model has Anderson localization for all  $\theta$ \cite{figotin1984, simon85}.
 Indeed it was shown by Jitomirskaya-Liu  \cite{ref8} that  $\sigma_{p p}\left(H_{\lambda \tan, \alpha, \theta}\right)$ can be characterized arithmetically: by defining an index
 \begin{equation}\label{de}
\delta(\alpha, \theta):=\limsup _{n \rightarrow \infty} \frac{\ln q_{n+1}+\ln \left\|q_{n}\left(\theta-\frac{1}{2}\right)\right\|_{\mathbb{T}}}{q_{n}},
\end{equation}
Jitomirskaya-Liu  \cite{ref8} show that
\begin{equation}\label{pp}\sigma_{p p}\left(H_{\lambda \tan, \alpha, \theta}\right)=\left\{E: L(E) \geq \delta(\alpha, \theta)\right\},\end{equation} while we have  $$\sigma_{s c}\left(H_{\lambda \tan, \alpha, \theta}\right)=\overline{\left\{E: L(E)<\delta(\alpha, \theta)\right\}}$$ where $L(E)$ is the Lyapunov exponent. 
 This makes the Maryland the  first model  where arithmetic spectral transition is described without
any parameter exclusion. 
It should be noted that  the proofs of localization in \cite{ref8}, as well as the original physics paper \cite{ref6}, is based on  a Cayley transform \cite{simon85} that reduced the eigenvalue problem to solving certain explicit cohomological equation.
In \cite{ref13}, the authors provided a constructive proof for the localization component  by expanding Jitomirskaya's localization approach \cite{ref11}, obtaining  Anderson localization for all $\theta$ and Diophantine $\alpha$. 
Quite recently,  Han-Jitomirskaya-Yang \cite{ref81} extended \cite{ref13}, gave a  constructive proof of \eqref{pp} for any 
 irrational $\alpha$.  More importantly, they investigated that, different from AMO, the Maryland model has another resonance: anti-resonance; this type of observation is critical in proving the arithmatic transition. The natural question is whether there are other quasi-periodic unbounded models that exhibit arithmetic  phase transitions, and whether the monotonicity is an essential assumption.

To answer these questions, we study the following unbounded Schr\"odinger operator:
\begin{equation}
\left(H_{V_1, \alpha, \theta} u\right)_{n}=u_{n+1}+u_{n-1}+\lambda V_1(\theta+\frac{n \alpha}{2},n) u_{n}=E u_{n},
\label{123}
\end{equation}
where $$
V_1(\theta, n)=\left\{\begin{array}{cc}
  \tan  \pi \theta, & n \in 2 \mathbb{Z}, \\
0, & \text { else }.
\end{array}  \right.
$$
Be aware that this potential exhibits strong oscillations, we refer to it as the Mosaic Maryland operator. The name of the operator was inspired by a recently researched quasi-periodic mosaic model\cite{ref22,wang2020}:
 \begin{equation*}
	(H_{V_{2},\alpha,\theta}u)_n=u_{n+1}+u_{n-1}+V_{\theta}(n)u_{n},
\end{equation*}
where
\begin{equation*}\quad
	V_{\theta}(n)=\left\{\begin{matrix}2\lambda \cos2\pi(\theta+n\alpha),&n\in \kappa \mathbb{Z},\\ 0,&else,\end{matrix}\right.\quad\lambda>0.
\end{equation*}
and the authors demonstrate the existence of exact mobility edges \cite{ref22}, which are energies separating absolutely continuous spectrum from pure point spectrum. For the mosaic Maryland  model, we show the following phase transition result:
\begin{Theorem}\label{main1}
 Let $\alpha \in \mathbb{R} \backslash \mathbb{Q}$,  then Lyapunov exponent of $H_{V_1, \alpha, \theta}$ satisfy 
 $$ L(E)= arccosh (\frac{\sqrt{(E^2-4)^2+(\lambda E )^2}+\sqrt{(E^4+(\lambda E)^2)}}{4}).$$
 Moreover, $H_{V_1, \alpha, \theta}$
     has purely singular continuous spectrum on $\left\{E:0< L(E)<\delta(\alpha, \theta)/2\right\}$, and pure point spectrum on $\left\{E: L(E)>\delta(\alpha, \theta)/2\right\}$.
\end{Theorem}

Let's give some comments why Theorem \ref{main1} is interesting. While  Cayley transform \cite{simon85} can be used to prove pure point part of the Maryland model, it doesn't work the mosaic Maryland model, thus to prove the Anderson localization part of Theorem \ref{main1}, we have to adopt the constructive proof \cite{ref81,ref13}.
Note that for quasi-periodic unbounded models, if the potential is monotonic and the frequency is Diophantine, there are essentially no resonances, making localization proof relatively simple, this can be seen either  from the KAM side \cite{bellissard1983} or from the Green's function estimation side\cite{ref13}. In our non-monotonic model, our proof follows from  \cite{ref81}, 
and we will further explore the anti-resonances lead to Anderson localization.
From the singular continuous side, the proof will based on sharp Gordon's argument  \cite{ayz,ref8}. To the best knowledge of the authors, we  present the first quasi-periodic unbounded  model with non-monotonic model, that has arithmetic phase transition.

We also note the Mosaic Maryland operator is generated by product systems, which corresponds to a periodic multiplicative modulation of Maryland potential. Clearly then, $V_{1}(\theta, n)$ admits a description in terms of the product system $X=\mathbb{T} \times \mathbb{Z}_{2} $, $T: X \rightarrow X,(\theta, n) \mapsto (\tilde{\alpha}+\theta, n+1)$. In particular,
$$
V_{1}(n, \theta )=V_{\omega}(n)=f\left(T^{n} \omega\right).
$$
where $\omega=\left(n, \theta\right)$ and
$$
f\left(n, \theta\right)=\tan(\pi(\theta))f_{2}(n),
$$
with $f_{2}(n)=\delta_{n \bmod  2, 0}$.
Indeed, it is a special case of ergodic Schr\"{o}dinger  operators defined over product dynamical systems in which one factor is periodic and the other factor is either a subshift over a finite alphabet or an irrational rotation of the circle. We point the reader to \cite{ref21} for a thorough account of spectral properties of dynamically defined Schr\"{o}dinger operators.


\section{Preliminaries}

\subsection{Rational approximations}
Let $\alpha \in(0,1) \backslash \mathbb{Q}, a_{0}=0$, and let $\alpha_{0}=\alpha$. Inductively for $k \geq 1$,
$$
a_{k}=\left[\alpha_{k-1}^{-1}\right], \quad \alpha_{k}=\alpha_{k-1}^{-1}-a_{k}=G\left(\alpha_{k-1}\right)=\left\{\frac{1}{\alpha_{k-1}}\right\} .
$$
Let $p_{0}=0, p_{1}=1, q_{0}=1, q_{1}=a_{1}$, then we define inductively $p_{k}=a_{k} p_{k-1}+$ $p_{k-2}, q_{k}=a_{k} q_{k-1}+q_{k-2}$. The sequence $\left(q_{n}\right)$ are the denominators of the best rational approximations of $\alpha$, since we have
\begin{equation}
\forall 1 \leq k<q_{n}, \quad\|k \alpha\|_{\mathrm{T}} \geq\left\|q_{n-1} \alpha\right\|_{\mathrm{T}},
\end{equation}
and
\begin{equation}
\frac{1}{2 q_{n+1}} \leq\left\|q_{n} \alpha\right\| \leq \frac{1}{q_{n+1}},
\label{111}
\end{equation}
\begin{equation}
\left\|q_{n-1} \alpha\right\|=a_{n+1}\left\|q_{n} \alpha\right\|+\left\|q_{n+1} \alpha\right\|.
\end{equation}

\subsection{Cocycles and growth of the cocycle}

Let $X$ be a compact metric space, $(X, \nu, T)$ be ergodic. A cocycle $(\alpha, A) \in \mathbb{R} \backslash \mathbb{Q} \times C^{\omega}(X, M(2, \mathbb{R}))$ is a linear skew product:
$$
\begin{gathered}
(T, A): \quad X \times \mathbb{R}^{2} \rightarrow X \times \mathbb{R}^{2} \\
(x, \phi) \mapsto(T x, A(x) \cdot \phi).
\end{gathered}
$$
For $n \in \mathbb{Z},  A_{n}$ is defined by $(T, A)^{n}=\left(T^{n}, A_{n}\right)$. Thus $A_{0}(x)=i d$,
$$
A_{n}(x)=\prod_{j=n-1}^{0} A\left(T^{j} x\right)=A\left(T^{n-1} x\right) \cdots A(T x) A(x), \text { for } n \geq 1.
$$
and $A_{-n}(x)=A_{n}\left(T^{-n} x\right)^{-1}$. $A_{n}$ is called the $n$-step transfer matrix. For this kind of cocycles, the Lyapunov exponent
$$
L(\alpha, A)=\lim _{n \rightarrow \infty} \frac{1}{n} \int \ln \left\|A_{n}(\theta)\right\| d \theta
$$
is well defined.
In this paper, we will consider the following  cocycle:
 $X=\mathbb{T} \times \mathbb{Z}_{2}$ and $T=T_{\alpha}$, where $ T_{\alpha}(\theta, n)=(\theta+\alpha/2, n+1)$, then $\left(T_{\alpha}, A\right)$ defines an almost-periodic cocycle.
These dynamical system $(X, T)$ is uniquely ergodic if $\alpha$ is irrational \cite{ref9}.

Consider the quasi-periodic Schrodinger equation :
\begin{equation}
\left(H_{V, \alpha, \theta} u\right)_{n}=u_{n+1}+u_{n-1}+V(\theta+n \alpha) u_{n}=E u_{n},
\label{234}
\end{equation}
then  the Schrodinger cocycle$(\alpha,S_{E}^V)$ is defined as
\begin{equation*}
S_{E}^{V}(\cdot)=\left(\begin{array}{cc}
E-V(\cdot) & -1 \\
1 & 0
\end{array}\right), \quad E \in \mathbb{R}.
\end{equation*}
Thus, any (formal) solution $\phi$ of (\ref{234})  can be reconstructed via the following relation

$$
\left(\begin{array}{c}
\phi(k+1) \\
\phi(k)
\end{array}\right)=S_E^{V}(\theta+k\alpha)\left(\begin{array}{c}
\phi(k) \\
\phi(k-1)
\end{array}\right).
$$

\subsection{Trigonometric product}
The following lemma from \cite{ref2} gives a useful estimate of
products appearing in our analysis.

\begin{Lemma}\label{av31}
 For $\alpha \in \mathbb{R} \backslash \mathbb{Q}, \theta \in \mathbb{R}$ and $0 \leq j_{0} \leq q_{n}-1$  be such that
$$
\left|\cos \pi\left(\theta+j_{0} \alpha\right)\right|=\inf _{0 \leq j \leq q_{n}-1}|\cos \pi(\theta+j \alpha)|,
$$
then for some absolute constant $C >0$
$$
-C \ln q_{n} \leq \sum_{j=0, j \neq j_{0}}^{q_{n}-1} \ln |\cos \pi(\theta+j \alpha)|+\left(q_{n}-1\right) \ln 2 \leq C \ln q_{n}.
\label{151}
$$
\end{Lemma}

\section{Lyapunov exponents} \label{le}
  \quad To exactly calculate the Lyapunov exponent, we need to consider $L(\alpha, A(\cdot+i\epsilon))$ with complex
phase $\varepsilon$. The basic idea is to reduce the non-trival problem of computing the Lyapunov exponent of a given non-constant cocycle
to an "almost constant" cocycle by taking $\varepsilon \rightarrow \infty $. This approach was first
developed by Avila.
 
Let us make a short review of Avila's global theory of one-frequency quasi-periodic cocycles \cite{ref1}. Suppose that $D \in C^{\omega}(\mathbb{T}, M(2, \mathbb{C}))$ admits a holomorphic extension to$\{|\Im \theta|<h\}$. Then for $|\epsilon|<h$, we define $D_{\epsilon} \in C^{\omega}(\mathbb{T}, M(2, \mathbb{C}))$ by $D_{\epsilon}(\cdot)=S_{E}^{V}(\cdot+i \epsilon)$, and define the the acceleration of $\left(\alpha, D_{\varepsilon}\right)$ as follows
$$
\omega\left(\alpha, D_{\varepsilon}\right)=\frac{1}{2 \pi} \lim _{h \rightarrow 0+} \frac{L\left(\alpha, D_{\varepsilon+h}\right)-L\left(\alpha, D_{\varepsilon}\right)}{h}.
$$
The acceleration was first introduced by Avila for analytic $S L(2, \mathbb{C})$ cocycles \cite{ref1}, and extended to analytic $M(2, \mathbb{C})$ cocycles by Jitomirskaya and Marx \cite{ref24,ref25}. It follows from the convexity and continuity of the Lyapunov exponent that the acceleration is an upper semicontinuous function in parameter $\varepsilon$. The key property of the acceleration is that it is quantized:

\begin{Theorem}\label{global}
Suppose that $(\alpha, D) \in$ $(\mathbb{R} \backslash \mathbb{Q}) \times C^{\omega}\left(\mathbb{T}, M_{2}(\mathbb{C})\right)$ with det $D(\theta)$ bound away from 0 on the strip $\{| \Im \theta| <h\}$, then $\omega\left(\alpha, D_{\varepsilon}\right) \in \frac{1}{2} \mathbb{Z}$ in the strip. Morveover, if $D \in C^{\omega}(\mathbb{T}, S L(2, \mathbb{C}))$, then $\omega\left(\alpha, D_{\varepsilon}\right) \in \mathbb{Z}$. 
\end{Theorem}

Now, we consider the Lyapunov exponent of the  model defined in (\ref{123}). $V_{1}$ is defined on $\mathbb{T} \times \mathbb{Z}_{2}$, consequently (\ref{123}) induces an almost-periodic Schrödinger cocycle $\left(T_{\alpha}, S_{E}^{V_{1}}\right)$ where $T_{\alpha}(\theta, n)=(\theta+\alpha/2, n+1)$. Although $\left(T_{\alpha}, S_{E}^{V_{1}}\right)$ is not a quasi-periodic cocycle in the strict sense, its iterate
$$
\left( \alpha, D_{E}^{V_{1}}\right)=:\left(\alpha, S_{E}^{V_{1}}(\theta, 1) \times S_{E}^{V_{1}}(\theta, 0)\right),
$$
indeed defines an analytic quasi-periodic cocycle. By simple calculation,
$$
\begin{aligned}
D_{E}^{V_{1}}(\theta) &=\left(\begin{array}{cc}
E & -1 \\
1 & 0
\end{array}\right)\left(\begin{array}{cc}
E- \lambda \tan \pi \theta & -1 \\
1 & 0
\end{array}\right) \\
&=\left(\begin{array}{ccc}
E^{2} & - \lambda E \tan  \pi \theta-1 & -E \\
& E- \lambda \tan  \pi \theta & -1
\end{array}\right).
\end{aligned}
$$
It is easy to see that $L\left(T_{\alpha}, S_{E}^{V_{1}}\right)=\frac{1}{2} L\left(\alpha, D_{E}^{V_{1}}\right)$.
The latter can be explicitly computed by Avila’s global theory,
thus we have the following result:
\begin{Lemma}\label{le2}
For  $ \alpha \in \mathbb{R} \backslash \mathbb{Q}$ and $\lambda \in \mathbb{R}$, we have
$$
4cosh(L\left(T_{\alpha}, S_{E}^{V_{1}}\right))= \sqrt{(E^2-4)^2+(\lambda E )^2}+\sqrt{(E^4+(\lambda E)^2)}.
$$
\end{Lemma}
\begin{proof}
For simplicity, denote $L(E)=L\left(T_{\alpha}, S_{E}^{V_{1}}\right) $. It suffices for us to prove that for any  $E\in \Sigma(H_{V_1, \alpha, \theta})$, we have

$$4cosh(\frac{1}{2} L\left( \alpha, D_{E}^{V_{1}}\right))= \sqrt{(E^2-4)^2+(\lambda E )^2}+\sqrt{(E^4+(\lambda E)^2)}.
$$
First we rewrite the matrix $D_{E}^{V_{1}}(\theta)$ as
$$
D_{E}^{V_{1}}(\theta)=\left(\begin{array}{cc}
E^{2}+i\lambda\frac{(e^{i 2 \pi \theta}-1)}{(e^{i 2 \pi \theta}+1)} E-1 & -E \\
E+i\lambda \frac{(e^{i 2 \pi \theta}-1)}{(e^{i 2 \pi \theta}+1)} & -1
\end{array}\right),
$$
then we complexify the phase
$$
D_{E}^{V_{1}}(\theta+i\epsilon)=\left(\begin{array}{cc}
E^{2}+i\lambda\frac{(e^{i 2 \pi (\theta+i\epsilon)}-1)}{(e^{i 2 \pi (\theta+i\epsilon)}+1)} E-1 & -E \\
E+i\lambda \frac{(e^{i 2 \pi (\theta+i\epsilon)}-1)}{(e^{i 2 \pi (\theta+i\epsilon)}+1)} & -1
\end{array}\right).
$$
Let $\epsilon$ goes to infinity, then
$$
D_{E}^{V_{1}}(\theta+i \epsilon)=D_{\infty}+o(1),
$$
where
$$
D_{\infty}=\left(\begin{array}{cc}
E & -1 \\
1 & 0
\end{array}\right)\times \left(\begin{array}{cc}
E-i \lambda & -1 \\
1 & 0
\end{array}\right).
$$
By the continuity of the LE \cite{ref25,ref26}
$$ L(\alpha,D_{E}^{V_{1}}(\theta+i \epsilon))=L(\alpha ,D_{\infty} )+o(1).$$
The quantization of acceleration  yields
$$ L(\alpha,D_{E}^{V_{1}}(\theta+i \epsilon))=L(\alpha ,D_{\infty} ).$$
for all $\epsilon >0$ sufficiently large. In addition the convexity, continuity and symmetry of $ L(\alpha,D_{E}^{V_{1}}(\theta+i \epsilon))$ with respected to $\epsilon$, gives
$$ L(\alpha,D_{E}^{V_{1}}(\theta+i \epsilon))=L(\alpha ,D_{\infty} ) ,$$
for all $\epsilon >0$.
Note that symmetry means $ L(\alpha,D_{E}^{V_{1}}(\theta+i \epsilon))=L(\alpha,D_{E}^{V_{1}}(\theta-i \epsilon))$, this implies
$$L(E)=L(\alpha ,D_{\infty} )/2 .$$
Then Lemma \ref{le2}  follows from solving for the eigenvalue of $D_{\infty
}$(a constant matrix) directly.
\end{proof}
It is obviously that $L(E)=0$ if and only if $E=0$. Now,  we will prove:
\begin{Lemma}
We have, $0\in \Sigma(H_{V_1, \tilde{\alpha}, \theta})$ 
\end{Lemma}
\begin{proof}
Let 
\begin{equation}
u_{n}=\left\{\begin{array}{cc}
1 & n=4 k+1 \\
-1 & n=4 k+3 \\
0 & \text { else }
\end{array}\right.
\end{equation}
where $ k\in \mathbb{Z}$, direct computation shows  the sequence $\left(u_{n}\right)_{n \in \mathbb{Z}}$ satisfy the equation in (\ref{123}) when $E=0$. By Schnol’s theorem\cite{ref27}, $0\in \Sigma(H_{V_1, \alpha, \theta})$.

\end{proof}

\begin{remark}
In the remaining of the paper, we only consider the case  energy $E \in \Sigma(H_{V_1, \alpha, \theta})$ with positive Lyapunov exponent.
\end{remark}

\section{Singular Continuous Spectrum}

  Denote $A(\theta)=D_{E}^{V_1}(\theta)\times \cos(\theta)$
and
\begin{equation}\label{an}
\begin{aligned}
A_{m}(\theta) &=A(\theta+(m-1) \alpha) \cdots A(\theta+\alpha) A(\theta), \\
&=A^{m}(\theta) \cdots A^{2}(\theta) A^{1}(\theta).
\end{aligned}
\end{equation}
Let $B(\theta)=S_{E}^{V}(\theta)$ and
$$
\begin{aligned}
B_{m}(\theta) &=B(\theta+(m-1) \tilde{\alpha}) \cdots B(\theta+\tilde{\alpha}) B(\theta), \\
&=B^{m}(\theta) \cdots B^{2}(\theta) B^{1}(\theta).
\end{aligned}
$$
for $m\geq 1$ and $\tilde{\alpha}=\alpha/2$. We also denote $B_{-m}(\theta)=B_{m}(\theta-m \tilde{\alpha})^{-1}$. Then, we have the following

\begin{proposition}\label{p81}
If $E\in$ $\left\{E: 0<2L(E) < \delta(\alpha, \theta)\right\}$, there exists $N=N(E, \lambda, \epsilon)>0$ such that if $q_{n_i}>N$, let $\varphi(k)$ be a normalized solution of (\ref{123}), $\bar{u}_{E}^{\theta}=\left(\begin{array}{c}\varphi(0) \\ \varphi(-1)\end{array}\right)$,then we have
\begin{equation}
\ \left\|(B_{2q_{n_i}}\left(\theta+2q_{n_i} \tilde{\alpha}\right)-B_{2q_{n_i}}(\theta))\bar{u}_{E}^{\theta}\right\| \leq \ e^{(2L-\delta(\alpha) +4\epsilon) q_{n_i}},
\label{81}
\end{equation}
\begin{equation}
\ \left\|(B_{- 2q_{n_i}}\left(\theta+2q_{n_i} \tilde{\alpha}\right)-B_{- 2q_{n_i}}(\theta))\bar{u}_{E}^{\theta}\right\| \leq \ e^{(2L-\delta(\alpha) +4\epsilon) q_{n_i}}.
\label{82}
\end{equation}
\end{proposition}
\begin{proof}
We only give the proof of (\ref{81}), the proof of (\ref{82}) is similar. Note $B_{ 2q_{n}}(\theta)=\frac{A_{ q_{n}}(\theta)}{\prod_{j=0}^{q_n-1} \cos \pi(\theta +j \alpha)}={\prod_{j=0}^{q_n-1}\frac{A^{ j}(\theta)}{c_j(\theta)}}$, where $A^{ j}(\theta)=A(\theta+j\alpha)$, $c_j=\cos(\pi (\theta+j\alpha))$.  
By telescoping argument (One can consult \cite{2017s} for details), we have
$$
\begin{aligned}
&\quad\left\|(B_{ 2q_{n_i}}\left(\theta+2q_{n_i} \tilde{\alpha}\right)-B_{ 2q_{n_i}}(\theta))\bar{u}_{E}^{\theta}\right\|\\
& \leq \sum_{j=0}^{q_{n_i}-1}\left\|\left(\prod_{l=0}^{j-1} \frac{A^{q_{n_i}+l}}{c_{q_{n_i}+l}}\right)\left(\frac{A^{q_{n_i}+j}-A^{j}}{c_{q_{n_i}+j}}\left(\begin{array}{c}
\varphi_{j-1} \\
\varphi_{j-2}
\end{array}\right)-\frac{c_{q_{n}+j}-c_{j}}{c_{q_{n}+j}}\left(\begin{array}{c}
\varphi_{j} \\
\varphi_{j-1}\end{array}\right)\right) \right\| .\\
\label{89}
\end{aligned}
$$
Since $ \varphi \in \ell^{2}$ is decaying solution, there exists a constant $C > 0 $ such that
$$
\left\|\left(\begin{array}{c}
\varphi_{k} \\
\varphi_{k-1}
\end{array}\right)\right\| \leq C.
$$
And we  need to estimate the norms $A_{q_{n_i}+j-1}$. The following control of the norm of the transfer matrix of a uniquely ergodic continuous cocycle by the Lyapunov exponent is well known.
\begin{Theorem}(\cite{ref10,ref22})\label{uper}
Let $(\alpha, M) $be a continuous cocycle, then for any $ \varepsilon> 0$, for $|n| $ large
enough,$$
\left\|M_{n}(\theta)\right\| \leq e^{|n|(L(\alpha, M)+\varepsilon)} \text { for any } \theta \in \mathbb{T}.
$$
\end{Theorem}
 Since $A(\theta)=D_{E}^{V_1}(\theta)\times \cos(\theta)$ is analytic, we have that $\ln \left\|A_{n}(\theta)\right\|$ is a continuous subadditive cocycle,  by Theorem \ref{uper}, we have 
\begin{equation}\label{uper2}
\left\|A_{n}(\theta)\right\| \leq e^{|n|(L(\alpha, A)+\varepsilon)} \text { for any } \theta \in \mathbb{T},
\end{equation}
for any $ \varepsilon> 0$, for $|n| $ large. And by the fact that  $\int_{\mathbb{T}} \ln |\cos \pi \theta| d \theta=-\ln 2$, we have
$$ L(\alpha,A)=2L(E)-\ln 2.
$$
 Considering 1-dimensional continuous cocycles,  by Theorem \ref{uper},  we have the following  corollary.
\begin{corollary}[\cite{ref81}]\label{cos} 
Let $I=\left[\ell_{1}, \ell_{2}\right] \subset \mathbb{Z}$, we have
\begin{equation*}
\prod_{\ell=\ell_{1}}^{\ell_{2}}|\cos (\pi(\theta+\ell \alpha))| \leq C(\varepsilon) e^{\left(\ell_{2}-\ell_{1}\right)(-\ln 2+\varepsilon)} \inf _{j=\ell_{1}}^{\ell_{2}}|\cos (\pi(\theta+j \alpha))|,
\end{equation*}
where $C(\varepsilon)$ is a constant that depends only on $\varepsilon$.
\end{corollary}
As for the lower bound of $\prod_{j} c_{j}$,  we will use the following Lemma.  
\begin{Lemma}[Theorem 2.3, \cite{ref8}]\label{cos3}
 For any $\epsilon>0$, there exists a sub sequence $q_{n_{i}}$ of $q_{n}$ such that the following estimate holds
\begin{equation}\label{4.6}
\prod_{j=0}^{\bar{q}_{n_{i}}-1}\left|c_{j}\right| \geq \frac{e^{(\delta(\alpha, \theta)-\ln 2-\epsilon) \tilde{q}_{n_{i}}}}{\tilde{q}_{n_{i}+1}}.
\end{equation}
\end{Lemma}
Observe that $\sup _{\theta \in \mathbb{T}}\left\|A_{\pm q_{n_i}}\left(\theta+2q_{n_i} \tilde{\alpha}\right)-A_{\pm q_{n}}(\theta)\right\| \leq \frac{C}{q_{n_{i}+1}}$, combining  (\ref{uper2}), Corollary \ref{cos} with Lemma \ref{cos3}, we have
$$
\begin{aligned}
&\left\|\left(B_{2q_{n}}(\theta+2q_{n}\tilde{\alpha})-B_{2q_{n_i}}\left(\theta\right)\right)\left(\begin{array}{c}
\varphi_{0} \\
\varphi_{-1}
\end{array}\right)\right\| \\
\leq & C \frac{q_{n_{i}} e^{q_{n_{i}}(2L(E)-\ln{2}+\epsilon)} \cdot e^{q_{n_{i}} \epsilon}}{e^{q_{n_{i}}(\delta(\alpha)-\ln{2}-\epsilon)}} \\
\leq & e^{q_{n_{i}}(2L(E)-\delta(\alpha)+4 \epsilon)}.
\end{aligned}
$$
 
\end{proof}
As a result of Proposition \ref{p81}, we have the following:
\begin{corollary}\label{sc}
let $\varphi(k)$ be a normalized solution of (\ref{123}), $\bar{u}_{E}^{\theta}=\left(\begin{array}{c}\varphi(0) \\ \varphi(-1)\end{array}\right)$, then we have
\begin{equation}
\max \left\{\left\|B_{2q_{n_i}}(E, \theta) \bar{u}_{E}^{\theta}\right\|,\left\|B_{-2q_{n_i}}(E, \theta) \bar{u}_{E}^{\theta}\right\|,\left\|B_{4 q_{n_i}}(E, \theta) \bar{u}_{E}^{\theta}\right\|\right\} \geq \frac{1}{4}.
\end{equation}
\end{corollary}

\begin{proof}
 The proof is essentially contained in Lemma 3.2 of \cite{ref23}. We remark that this result
is only valid in the subsequence $n_{i}$.
\end{proof}

Now as a result of Corollary \ref{sc}, one can conclude that $H_{V_1, \alpha, \theta}$ has purely singular continuous spectrum on $\left\{E: 0< L(E)<\delta(\alpha, \theta)/2\right\}$.

\section{Pure Point Spectrum}
In this section, we are devote to prove  Anderson localization in the  regime $\left\{E: L(E)>\delta(\alpha, \theta)/2\right\}$.  We first introduce some notations and recall the key framework,  modified from the one developed in \cite{ref81, ref84} also with adaptions from \cite{ref13, ref22}. 
For any generalized eigenvalue $E$,  assume $\phi$ is the corresponding generalized eigenfunction of  $H_{V_1, \alpha, \theta}$, without loss of generality
assume 
\begin{equation}\label{eigen1}
|\phi(0)|\geq 1,
\end{equation} 
and
\begin{equation}
|\phi(k)| \leq C_{0}|k|.
\label{eigen}
\end{equation}
We shall write $\delta(\alpha, \theta)$ as $\delta$ and $\beta(\alpha)$ as $\beta$ for simplicity.
Define
\begin{equation}\label{a53}
\beta_{n}:=\frac{\ln q_{n+1}}{q_{n}},
\end{equation}
and
\begin{equation} \label{a54}
\delta_{n}:=\frac{\ln \left\|q_{n}\left(\theta-\frac{1}{2}\right)\right\|-\ln \left\|q_{n} \alpha\right\|}{q_{n}},
\end{equation}
then one can check that
\begin{Lemma}\cite{ref81}\label{delta} 
 We have $0 \leq \delta \leq \beta$ for all $\alpha, \theta$, and
$\delta=\lim \sup \max \left(0, \delta_{n}\right)$.
\end{Lemma}
Fix a small $\varepsilon>0$ such that 
\begin{equation}\label{lle}
 2L(E)>\delta+700 \varepsilon.
\end{equation}
Since $\limsup _{n \rightarrow \infty} \delta_{n}=\delta$, we have that for $n>N(\varepsilon)$ large enough,
\begin{equation}\label{5.5}
2L(E)>\delta_{n}+680 \varepsilon.
\end{equation}
Then we have the following:
\begin{Theorem}\label{local2}
If $E\in$ $\left\{E: 2L(E) > \delta(\alpha, \theta)\right\}$,   let $\phi$ be an generalized eigenfunction satisfying $|\phi(0)| \geq 1$ and (\ref{eigen}). Then for $n>N\left(\alpha, E, \lambda, \varepsilon, C_{0}\right)$ large enough and $\frac{1}{6} q_{n} \leq|k|<\frac{1}{6} q_{n+1}$, we have
$$
|\phi(k)| \leq e^{-\left(L-\delta_{n}(\alpha)/2-330 \varepsilon\right)|k|}.
$$
\end{Theorem}

Before giving  the proof, we first introduce some useful notations and concepts.
Denote by $M_{k}(\theta)$ the $k-$ step transfer-matrix of $H_{V, \alpha, \theta} u=E u$, and denote
$$
Q_{k}(\theta)=\operatorname{det}\left[\left.\left(H_{V_1, \alpha, \theta}-E\right)\right|_{[0, k-1]}\right], \quad P_{k}(\theta)=\operatorname{det}\left[\left.\left(H_{V_1, \alpha, \theta}-E\right)\right|_{[1, k]}\right],
$$
for $k\geq1$, then the $k$-step transfer-matrix can be written as
$$
M_{k}(\theta)=(-1)^{k}\left(\begin{array}{cc}
Q_{k}(\theta) & P_{k-1}(\theta) \\
-Q_{k-1}(\theta) & -P_{k-2}(\theta)
\end{array}\right).
$$

Let $\tilde{Q}_{k}(\theta): \mathbb{R} /  \mathbb{Z} \rightarrow \mathbb{R}$ be defined as $\tilde{Q}_{2k}(\theta)=\prod_{j=0}^{k-1} \cos \pi(\theta+j \alpha) \cdot Q_{2k}(\theta)$ and $\tilde{Q}_{2k+1}(\theta)=\prod_{j=0}^{k} \cos \pi(\theta+j \alpha) \cdot    Q_{2k+1}(\theta)$. Respectively,  $\tilde{P}_{k}(\theta)$ can be  also defined as $\tilde{P}_{2k}(\theta)=\prod_{j=1}^{k} \cos \pi(\theta+j \alpha) \cdot P_{2k}(\theta)$ and $\tilde{P}_{2k+1}(\theta)=\prod_{j=1}^{k} \cos \pi(\theta+j \alpha) \cdot P_{2k+1}(\theta)$.  Then clearly,  it turns out $A_{k}(\theta)$ defined in (\ref{an}) has the following expression
\begin{equation}\label{tr1}
A_{k}(\theta)=M_{2k}(\theta)\prod_{j=0}^{k-1} \cos \pi(\theta+j \alpha)=\left(\begin{array}{cc}
\tilde{Q}_{2k}(\theta) & -\tilde{P}_{2k-1}(\theta+\alpha) \cos \pi \theta \\
-\tilde{Q}_{2k-1}(\theta)  & -\tilde{P}_{2k-2}(\theta+\alpha) \cos \pi \theta 
\end{array}\right) 
\end{equation}

Then, we have the following upper bound of  $\tilde{P}_{k}$ and $\tilde{Q}_{k}$. 
\begin{Lemma} For any $ \epsilon> 0$, for $|k| $ large
enough,
\begin{equation}\label{q1}
\left|\tilde{P}_{k}(\theta)\right| \leq e^{(\tilde{L}(E)+\epsilon)|k|} \text { for any } \theta \in \mathbb{T}, 
\end{equation}
and
\begin{equation}\label{q2}
\left|\tilde{Q}_{k}(\theta)\right| \leq e^{(\tilde{L}(E)+\epsilon)|k|} \text { for any } \theta \in \mathbb{T}, 
\end{equation}
where $\tilde{L}(E)=L(E)-\frac{\ln 2}{2}$.
\end{Lemma}
\begin{proof}
It follows  from (\ref{uper2}) and (\ref{tr1}).
\end{proof}
We can also have the average lower bound of $\tilde{P}_{k}$.
\begin{Lemma}\label{up0}
By Herman's subharmonic trick, one has
 \begin{equation}\label{av4}
\frac{1}{k} \int_{0}^{1} \ln \left|\tilde{P}_{k}(\theta)\right| \mathrm{d} \theta=\frac{1}{k} \int_{0}^{1} \ln \left|\tilde{P}_{k}(2 \theta)\right| \mathrm{d} \theta \geq \tilde{L}
\end{equation}
\end{Lemma}
The proof of this lemma  is modification of that of   Lemma 3.1 in \cite{ref13}. We will leave
it in the appendix. 
An important observation that makes our analysis possible is
\begin{Lemma}
$\frac{\tilde{P}_{2k-1}(\theta)} { \cos ^{K-1} (\pi \theta)}$ and $\frac{\tilde{P}_{2k}(\theta)}{  \cos ^{k} (\pi \theta)}$ can be expressed as a polynomial of degree $ k-1 $ and $k$  respectively in $tan\pi \theta$, namely,
\begin{equation}
\frac{\tilde{P}_{2k-1}(\theta)}{(\cos \pi \theta)^{k}} \triangleq g_{k-1}(\tan \pi \theta),
\label{04.1}
\end{equation}
\begin{equation}
\frac{\tilde{P}_{2k}(\theta)}{(\cos \pi \theta)^{k+1}} \triangleq f_{k}(\tan \pi \theta),
\label{4.1}
\end{equation}
where $g_{k-1}$ is a polynomial of degree $k-1$, respectively, $f_{k}$ is a polynomial of degree $k$.
\end{Lemma}
\begin{proof}
Note that $V_{1}(\theta, 2 n+1)=0$ and $V_{1}(\theta, n+2)=V_{1}(\theta+ \alpha, n)$. Then if we expand the determinant $\operatorname{det}\left[\left.\left(H_{V, \alpha, \theta}-E\right)\right|_{[1,2 k-1]}\right]$ and $\operatorname{det}\left[\left.\left(H_{V, \alpha, \theta}-E\right)\right|_{[1,2 k]}\right]$ by the last column, we have
$$
\begin{aligned}
P_{2 k-1}(\theta) =-&E P_{2 k-2}(\theta)-P_{2 k-3}(\theta), \\
P_{2 k}(\theta) =(\tan \pi(\theta+&k \alpha)-E)P_{2 k-1}(\theta)-P_{2 k-2}(\theta).
\end{aligned}
$$
Recall the definition of $\tilde{P_{k}}$, we have
\begin{equation}\label{dedai}
  \begin{aligned}
\tilde{P}_{2 k-1}(\theta) =-&E \tilde{P}_{2 k-2}(\theta)-\tilde{P}_{2 k-3}(\theta)\cos\pi(\theta+(k-1 )\alpha), \\
\tilde{P}_{2 k}(\theta) =(\tan(\pi(\theta+k \alpha))-&E)\cos(\pi(\theta+k \alpha))\tilde{P}_{2 k-1}(\theta)-\tilde{P}_{2 k-2}(\theta)\cos(\pi(\theta+k \alpha)),
\end{aligned}  
\end{equation}
then (\ref{04.1}) and (\ref{4.1}) follow from an induction, by using (\ref{dedai}). 
\end{proof}

By the Lagrange interpolation formula, for any set of $k + 1 $ distinct $\theta_i$'s in $(-1/2, 1/2)$,  we have the following convenient representation
\begin{equation}\label{318}
\begin{aligned}
\tilde{P}_{2k}(\theta)=(\cos \pi \theta)^{k} g_{k}(\tan \pi \theta) &=\sum_{i=0}^{k} \tilde{P}_{2k}\left(\theta_{i}\right) \frac{\prod_{l \neq i} \tan \pi \theta-\tan \pi \theta_{l}}{\prod_{l \neq i} \tan \pi \theta_{i}-\tan \pi \theta_{l}} \cdot \frac{\cos ^{k} \pi \theta}{\cos ^{k} \pi \theta_{i}} \\
&=\sum_{i=0}^{k} \tilde{P}_{2k}\left(\theta_{i}\right) \prod_{l \neq i} \frac{\sin \pi\left(\theta-\theta_{l}\right)}{\sin \pi\left(\theta_{i}-\theta_{l}\right)},
\end{aligned}
\end{equation}
also
\begin{equation}
\begin{aligned}
\tilde{P}_{2k+1}(\theta)=(\cos \pi \theta)^{k} g_{k}(\tan \pi \theta) &=\sum_{i=0}^{k} \tilde{P}_{2k-1}\left(\theta_{i}\right) \frac{\prod_{l \neq i} \tan \pi \theta-\tan \pi \theta_{l}}{\prod_{l \neq i} \tan \pi \theta_{i}-\tan \pi \theta_{l}} \cdot \frac{\cos ^{k} \pi \theta}{\cos ^{k} \pi \theta_{i}} \\
&=\sum_{i=0}^{k} \tilde{P}_{2k+1}\left(\theta_{i}\right) \prod_{l \neq i} \frac{\sin \pi\left(\theta-\theta_{l}\right)}{\sin \pi\left(\theta_{i}-\theta_{l}\right)}.
\label{319}
\end{aligned}
\end{equation}

In this regard, we also recall the following useful concept:
\begin{definition}\cite{ref2}
We say that the set $\left\{\theta_{1}, \ldots, \theta_{k+1}\right\}$ is $\epsilon$ -uniform if
\begin{equation}
\max _{\theta \in[0,1]} \max _{i=0, \ldots, k} \prod_{l \neq i} \frac{\left|\sin \pi\left(\theta-\theta_{l}\right)\right|}{\left|\sin \pi\left(\theta_{i}-\theta_{l}\right)\right|}<e^{k \epsilon}.
\label{41}
\end{equation}
\end{definition}

We use $G_{\left[x_{1}, x_{2}\right]}(E)(x, y)$ for the Green function of the operator $H$ restricted to the interval $[x_1, x_2]$ with zero boundary conditions at $x_1 -1$ and $x_2 + 1$. We will omit $ E$ when it is fixed throughout the argument. A useful definition about Green’s function is the following:
\begin{definition}\cite{ref11}
 A point $y \in \mathbb{Z}$ will be called $(m, h)$-regular if there exists an interval $\left[x_{1}, x_{2}\right]$, $x_{2}=x_{1}+h-1$, containing $y$, such that $$
\left|G_{\left[x_{1}, x_{2}\right]}\left(x_{i}, y\right)\right|<e^{-m\left|y-x_{i}\right|}, \left|y-x_{i}\right| \geq \frac{1}{4} h, \text { for } i=1,2.
$$
Otherwise, $y \in \mathbb{Z}$ will be called $(m, h)$-singular.
\end{definition}
Let  $\phi(x)$  be a solution  of $H\phi(x)=E\phi(x)$ and let $[x_1, x_2] $ be an interval containing $y$. We have
\begin{equation}
\phi(y)=-G_{\left[x_{1}, x_{2}\right]}\left(x_{1}, y\right) \phi\left(x_{1}-1\right)-G_{\left[x_{1}, x_{2}\right]}\left(x_{2}, y\right) \phi\left(x_{2}+1\right).
\label{l1}
\end{equation}
In general, if $I=[a, b]$, let $\partial I:=\{a, b\}$ and $a^{\prime}:=a-1, b^{\prime}:=b+1$.  
If we denote
$$
\Delta_{m, n}(\theta)=\operatorname{det}\left[\left.\left(H_{V, \alpha, \theta}-E\right)\right|_{[m, n]}\right].
$$
By Cramer's rule, we have the following connection between the determinants ${P}_{k}$ and Green
function:
\begin{equation}
\begin{aligned}
&\left|G_{\left[x_{1}, x_{2}\right]}\left(x_{1}, y\right)\right|=\left|\frac{\Delta_{y+1, x_{2}}(\theta)}{\Delta_{x_{1}, x_{2}}(\theta)}\right|, \\
&\left|G_{\left[x_{1}, x_{2}\right]}\left(y, x_{2}\right)\right|=\left|\frac{\Delta_{x_{1}, y-1}(\theta)}{\Delta_{x_{1}, x_{2}}(\theta)}\right|.
\end{aligned}
\end{equation}
Furthermore, if $y=2n $ and $x_{i}=2n_{i}+1$  with  $n, n_{i}  \in  \mathbb{N}$ for $i=1, 2$, we have 
\begin{equation}
|\phi(y)| \leq \frac{\left|\tilde{P}_{x_{2}-y}\left(\theta_{n}\right)\right|}{\left|\tilde{P}_{x_{2}-x_{1}}\left(\theta_{n_{1}}\right)\right|} \prod_{k=n_{1} }^{n-1}\left|\cos \left(\pi \theta_{k}\right)\right| \cdot\left|\phi\left(x_{1}-1\right)\right|+\frac{\left|\tilde{P}_{y-x_{1}}\left(\theta_{n_{1}}\right)\right|}{\left|\tilde{P}_{x_{2}-x_{1}}\left(\theta_{n_{1}}\right)\right|} \prod_{k=n}^{ n_{2} }\left|\cos \left(\pi \theta_{k}\right)\right| \cdot\left|\phi\left(x_{2}+1\right)\right|,
\label{green}
\end{equation}
and if $y=2n $ and $x_{i}=2n_{i}$  with  $n, n_{i}  \in  \mathbb{N}$ for $i=1, 2$, we have 
\begin{equation}
|\phi(y)| \leq \frac{\left|\tilde{Q}_{x_{2}-y}\left(\theta_{n}\right)\right|}{\left|\tilde{Q}_{x_{2}-x_{1}}\left(\theta_{n_{1}}\right)\right|} \prod_{k=n_{1} }^{n-1}\left|\cos \left(\pi \theta_{k}\right)\right| \cdot\left|\phi\left(x_{1}-1\right)\right|+\frac{\left|\tilde{Q}_{y-x_{1}}\left(\theta_{n_{1}}\right)\right|}{\left|\tilde{Q}_{x_{2}-x_{1}}\left(\theta_{n_{1}}\right)\right|} \prod_{k=n}^{ n_{2} }\left|\cos \left(\pi \theta_{k}\right)\right| \cdot\left|\phi\left(x_{2}+1\right)\right|.
\label{green2}
\end{equation}
where $\theta_{k}=\theta +k \alpha$. One should be mentioned that if we use (\ref{green}) to expand even point $y$ with odd endpoints $x_1$ and $x_2$, then we can also expand even point $x^{\prime}_{i}$  in  a interval with odd endpoints. The parity of these points will contribute to keep the the numerators and denominators of Green’s functions from be replaced by $\tilde{Q}_{k}$.
\subsection{Key technical Lemmas }\label{ktl}

In the remaining of this paper, We want to prove the generalized eigenfunction $\phi$ decays exponentially (Theorem \ref{local2}). To do so, first we need to obtain good bound of $P_{k}$ and the product (indeed the minimum) of cosines in (\ref{green}). 
 Before giving these bounds of $P_{k}$ and the product  of cosines, we  need some concepts, which were first introduced in \cite{ref81}. 
\begin{definition}\cite{ref81}
 We call $(m, \ell) \in \mathbb{Z}^{2}$ is $\theta$-minimal on scale $q_{n}$ if the following holds

(1) $m \in\left[-q_{n} / 2, q_{n} / 2\right)$,

(2) $|\ell| \leq \frac{1}{q_{n}}\left(e^{\delta_{n} q_{n}}+q_{n}+\frac{1}{2}\right)$,

(3) $\left\|\theta-\frac{1}{2}+\left(m+\ell q_{n}\right)\right\|<\left(\frac{1}{2}+\frac{1}{2 q_{n}}\right)\left\|q_{n} \alpha\right\|$,

(4) (i). For $a_{n+1} \geq 4$, we have
$$
\left\|\theta-\frac{1}{2}+\left(m+j q_{n}\right) \alpha\right\| \leq 20 \min _{|k|<q_{n}}\left\|\theta-\frac{1}{2}+\left(m+j q_{n}+k\right) \alpha\right\|,
$$
holds for any $|j| \leq a_{n+1} / 6$.

  (ii). For $a_{n+1} \leq 3$, we have
$$ 
\left\|\theta-\frac{1}{2}+m \alpha\right\| \leq 20 \min _{-q_{n} / 2 \leq k<q_{n} / 2}\left\|\theta-\frac{1}{2}+k \alpha\right\|.
$$
\end{definition}

The following Lemma show  the existence of $\theta$-minimal $(m, \ell)$.
\begin{Lemma}\cite{ref81}\label{min}
 For any $q_{n}$ sufficiently large, there exists $\theta$- minimal $\left(m_{n}, \ell_{n}\right)$ at scale $q_{n}$.
\end{Lemma}

Define 
\begin{equation}\label{5.27}
 c_{n, \ell}:=\left|\cos \left(\pi \theta_{m_{n}+\ell q_{n}}\right)\right|.
\end{equation}
As a corollary of  Lemma \ref{min}, we have the following Lemma. 

\begin{Lemma}\label{cos2}
 Let $I=\left[\ell_{1}, \ell_{2}\right] \subset \mathbb{Z}$ be such that there exists $j \in \mathbb{Z},|j|<q_{n+1} /\left(6 q_{n}\right)$, that satisfies
$$
I \subset\left[m_{n}+j q_{n}+1, m_{n}+(j+1) q_{n}-1\right] .
$$
Then for $n>N(\varepsilon)$ large enough, we have
$$
\prod_{\ell \in I}\left|\cos \left(\pi \theta_{\ell}\right)\right| \geq e^{-\varepsilon\left(2 q_{n}-|I|\right)} e^{-(\ln 2)|I|}.
$$
Furthermore, for $\beta_{n} \geq \delta_{n}+200 \varepsilon$, for $|\ell| \leq q_{n+1} /\left(6 q_{n}\right)$ and some absolute constant $0<C<8$
 \begin{equation}\label{5.38}
c_{n, \ell} \leq C \max \left(|\ell|, e^{\delta_{n} q_{n}}, 1\right) e^{-\beta_{n} q_{n}}.
\end{equation}
\end{Lemma}
\begin{proof}
 One can consult  Corollary 5.5 and Corollary 5.6 of \cite{ref81} for details.
\end{proof}

Now, it is time to estimating $\tilde{P}_{k}$. 
\begin{Lemma}\label{up}
 Let $I_{1}, I_{2}$ be two disjoint intervals in $\mathbb{Z}$ such that $\left|I_{1} \cup I_{2}\right|=k$  and $\{\theta+\ell \alpha\}_{\ell \in I_{1} \cup I_{2}}$ is $\varepsilon_{k}$-uniform, then exists $x_{1}  \in I_{1} \cup I_{2} $ such that
$$
\left|\tilde{P}_{2 k-1}\left(\theta+x_{1} \alpha\right)\right| \geq e^{2 k\left(\tilde{L}-2 \varepsilon_{k}\right)}.
$$
\end{Lemma}
\begin{proof}
The result is direct consequence of the Lagrange interpolation formula and (\ref{av4}).  We omit the
details.
\end{proof}

Usually, the numerators of Green’s functions can be bounded uniformly by (\ref{q1}). Using the strategy in \cite{ref84} and  the inequality above one can prove $\phi(y)$ exponential decay. However this does not work for $\delta<2L(E)<\beta$, so one has
to look for an additional decay, which was first introduced in \cite{ref81}.  The following lemmas on $\tilde{P}_{k}$ are  essential for proving Anderson localization in the sharp regime  $\left\{E: L(E)>\delta(\alpha, \theta)/2\right\}$.

\begin{Lemma}\label{7.4}( Corollary 7.4, \cite{ref81})
 For $|\ell|<2 q_{n+1} /\left(3 q_{n}\right)$, assume $k<2 q_{n}$ and
$$
y \leq \ell q_{n}+m_{n}, \text { and } y+k-1 \geq(\ell+1) q_{n}+m_{n}-1,
$$
we then have
$$
\left|\tilde{P}_{2k-1}\left(\theta_{y}\right)\right| \leq g_{k, \ell} e^{(2k-1) \tilde{L}}.
$$
where
\begin{equation}\label{7.9}
g_{k, \ell}:= \begin{cases}\max \left(e^{\delta_{n} q_{n}},|\ell|, 1\right) e^{-\left(\beta_{n}-6 \varepsilon\right) q_{n}} & \text { if } \beta_{n} \geq \delta_{n}+200 \varepsilon \\ e^{2 \varepsilon k} & \text { if } \beta_{n}<\delta_{n}+200 \varepsilon\end{cases}.
\end{equation}
\end{Lemma}

\subsection{Some useful Propositions.}

Choose a value (from multiple possible values) of $\tau_{n}$ such that
$$
\tau_{n} \in\left(\frac{\varepsilon}{2 \max (L, 1)}, \frac{\varepsilon}{\max (L, 1)}\right]
$$
and $\tau_{n} q_{n} \in \mathbb{Z}$. Define $b_{n}=\tau_{n} q_{n}$. For any $m_{n} \in \mathbb{Z}$ we call $m_{n}$ resonant (at the scale of $q_{n}$ ) if $\operatorname{dist}\left(m_{n}, q_{n} \mathbb{Z}\right) \leq b_{n}$, otherwise we call $ y$ non-resonant.
we call $y$ even-resonant (at the scale of $q_{n}$ ) if $\operatorname{dist}\left(y, 2q_{n} \mathbb{Z}\right) \leq 2b_{n}$, otherwise we call $ y$ is  not even-resonant.
We introduce some notations:
$$
\left\{\begin{array}{l}
I^{-}:=\left[2\ell q_{n}+2b_{n}, 2(\ell q_{n}+m_{n})-1\right] \\
I^{+}:=\left[2(\ell q_{n}+m_{n})+1,2(\ell+1) q_{n}-2b_{n}]\right. \\
\left|\phi\left(x_{0}^{-}\right)\right|:=\max _{y \in I^{-}}|\phi(y)|, \\
\left|\phi\left(x_{0}^{+}\right)\right|:=\max _{y \in I^{+}}|\phi(y)|, \\
R_{\ell}:=\left[2\ell q_{n}-2b_{n}, 2\ell q_{n}+2b_{n}\right] \\
r_{\ell}:=\max _{k \in R_{\ell}}|\phi(k)| .
\end{array}\right.
$$

In the following, we distinguish the proof according to $m_{n}$ is resonant or not. And  each part  can be divided into two cases depending on $y$ is even-resonant or not.

\begin{proposition}\label{pr1}
Assume $\operatorname{dist}\left(m_{n}, q_{n}\mathbb{Z}\right)>b_{n}$,

(1)If $y$ is not even resonance, we have:
 for $y=2(\ell q_{n}+m_{n})$,
\begin{equation*}
|\phi(y)| \leq e^{29 \varepsilon q_{n}} c_{n, \ell} \max \left(e^{-\left(y-2\ell q_{n}\right) L} r_{\ell}, e^{-\left(2(\ell+1) q_{n}-y\right) L} r_{\ell+1}\right).
\end{equation*}
For any $y \in  I^{-}$,
\begin{equation*}
|\phi(y)| \leq e^{29 \varepsilon q_{n}} \max \left(e^{-\left(y-2\ell q_{n}\right) L} r_{\ell}, c_{n, \ell} e^{-\left(2(\ell+1) q_{n}-y\right) L} r_{\ell+1}\right).
\end{equation*}
For any $y \in I^{+}$,
\begin{equation*}
|\phi(y)| \leq e^{29 \varepsilon q_{n}} \max \left(c_{n, \ell} e^{-\left(y-2\ell q_{n}\right) L} r_{\ell}, e^{-\left(2(\ell+1) q_{n}-y\right) L} r_{\ell+1}\right).
\end{equation*}

(2)If  $y$ is even-resonance,  we have:
for any $\ell \neq 0,|\ell| \leq q_{n+1} /\left(6 q_{n}\right)$,
\begin{equation}\label{r82}
r_{\ell} \leq  \frac{e^{-( 2L-55 \varepsilon) q_{n}}}{\max (|\ell|, 1)} \max \left(r_{\ell-1}, r_{\ell+1}\right) \times \begin{cases}\max \left(|\ell|, e^{\delta_{n} q_{n}}\right), & \text { if } \beta_{n} \geq \delta_{n}+200 \varepsilon \\ e^{\beta_{n} q_{n}}, & \text { if } \beta_{n}<\delta_{n}+200 \varepsilon\end{cases}.
\end{equation}
\end{proposition}

\begin{proposition}\label{pr2}
Assume $\operatorname{dist}\left(m_{n}, q_{n}\mathbb{Z}\right)\leq b_{n}$,

(1)If $y$ is not even resonance, we have: 
\begin{equation*}
|\phi(y)| \leq e^{40 \varepsilon q_{n}} \max \left(e^{-\left(y-2\ell q_{n}\right) L} r_{\ell}^{+}, e^{-\left(2(\ell+1) q_{n}-y\right) L} r_{\ell+1}^{-}\right) .
\end{equation*}
where 
$$
R_{\ell}^{+}:=\left[2\ell q_{n}+2m_{n}+1, 2\ell q_{n}+2b_{n}\right] \text { and } R_{\ell}^{-}:=\left[2\ell q_{n}-2b_{n}, 2\ell q_{n}+2m_{n}-1\right],
$$
and
$$
r_{\ell}^{+}:=\max _{y \in R_{\ell}^{+}}|\phi(y)| \text { and } r_{\ell}^{-}:=\max _{y \in R_{\ell}^{-}}|\phi(y)|.
$$

(2)If $y$ is  even resonance, for any $\ell \neq 0$ such that $|\ell|<q_{n+1} /\left(6 q_{n}\right)$, we have
\begin{equation}\label{r92}
r_{\ell} \leq  \frac{e^{-(2 L-70 \varepsilon) q_{n}}}{\max (|\ell|, 1)} \max \left(r_{\ell-1}, r_{\ell+1}\right) \times\left\{\begin{array}{l}
\max \left(|\ell|, e^{\delta_{n} q_{n}}\right), \text { if } \beta_{n} \geq \delta_{n}+200 \varepsilon \\
e^{\beta_{n} q_{n}}, \text { if } \beta_{n}<\delta_{n}+200 \varepsilon
\end{array}\right.
\end{equation}
\end{proposition}
  The above two propositions will be proved in Section \ref{0p0}. They will be used to prove  Theorem \ref{local2} in the case  $\beta_{n}$ is not too small. As for relevant Diophantine case , in other words, $0 \leq \beta_{n} \leq 300 \varepsilon$, we have the following: 
\begin{proposition}\label{pr3}
For $n$ large enough,

(1)If $\frac{q_{n}}{6}<k<q_{n}$, $k \in 2\mathbb{N}$, we have $|\phi(k)| \leq e^{-k(L-24 \varepsilon)}$.

(2)If $q_{n}<k<\frac{q_{n+1}}{6}$, $k \in 2\mathbb{N}$,  we have $|\phi(k)|  \leq e^{-(L-330 \varepsilon) k}$ .
\end{proposition}
It is a variant of case 1 of Lemma 6.1 in \cite{ref81}.We only need to replace $k$ by $k/2$ in the argument.

\subsection{Proofs of  Theorem \ref{local2}}
The remaining of this paper will be devoted to the proof of Theorem \ref{local2},
dividing into the following three cases.

Case 1. $ \beta_{n} \geq \max(\delta_{n}+200 \varepsilon, 300 \varepsilon)$;

Case 2. $ 300 \varepsilon  \leq \beta_{n}\leq \delta_{n}+200 \varepsilon $;

Case 3. $0 \leq \delta_{n} \leq \beta_{n} \leq 300 \varepsilon$. 

 Case 1  require some  key estimates   presented in Subsection \ref{ktl}. It is the most technical part in this paper as it showed in \cite{ref81}. In case 2, we have $2L>\beta_{n}+200 \varepsilon$, we will use  the strategy in \cite{ref84} to handle this case.  Compared to the Case 1,
Case 2 has a lot of simplifications. Case 3 is similarly to the Diophantine case that is handled in \cite{ref13}. 

\noindent\textbf {Case 1 }
Assume $\beta_{n} \geq \delta_{n}+200 \varepsilon$.  Let $y \in\left(2\ell q_{n}+2b_{n},2(\ell+1) q_{n}-2b_{n}\right)$ for some $|\ell| \leq \frac{q_{n}+1}{6 q_{n}}$. Without loss of generality, we assume $\ell \geq 0$.

If $\ell \neq 0,-1$,  we need the following Lemma:
\begin{Lemma}\label{10.1}
For any $\ell_{0}$ such that $1 \leq\left|\ell_{0}\right| \leq q_{n+1} /\left(6 q_{n}\right)$, we have
$$
r_{\ell_{0}} \leq e^{2\left(\delta_{n}/2-L+54 \varepsilon\right)\left|\ell_{0}\right| q_{n}}
$$
\end{Lemma}
\begin{proof}
In view of  (\ref{r82}) and (\ref{r92}), for any $0<\left|\ell_{0}\right| \leq q_{n+1} /\left(6 q_{n}\right)$, we have
\begin{equation}\label{iterate}
r_{\ell_{0}} \leq e^{\left(\delta_{n}/2-L+50 \varepsilon\right) 2q_{n}} \max _{\ell_{1}=\ell_{0} \pm 1} r_{\ell_{1}}.
\end{equation}
 One can iterate (\ref{iterate}) until one reaches $\ell_{t}$ (and stops the iteration once reaches such a $\ell_{t}$ ):

(1) $t=0$,

(2)$t=2 \ell_{0}$,

(3) the iterating number reaches $[q_{n+1} /\left(12 q_{n}\right)]$.

Hence one obtains
$$
r_{\ell_{0}} \leq \max _{\left(\ell_{0}, \ell_{1}, \ldots, \ell_{t}\right) \in \mathcal{G}} e^{\left(\delta_{n}/2-L+50 \varepsilon\right) 2t q_{n}} r_{\ell_{t}}
$$
where $\mathcal{G}=\left\{\left(\ell_{0}, \ldots, \ell_{t}\right):\left|\ell_{i}-\ell_{i-1}\right|=1\right.$

Then Lemma \ref{10.1} follows from bounding $\ell_{t}$ by (\ref{eigen}).
\end{proof}
Combing Proposition \ref{pr1}  (if $m_{n}$ is non-resonant) and \ref{pr2} (if $m_{n}$ is resonant) with Lemma \ref{10.1}, we have
\begin{equation}\label{e3}
|\phi(y)|  \leq e^{40 \varepsilon q_{n}} \max \left(e^{-\left(y-2\ell q_{n}\right) L} r_{\ell}, e^{-\left(2(\ell+1) q_{n}-y\right) L} r_{\ell+1}\right) 
\end{equation}

By (\ref{eigen}), we have 
\begin{equation}\label{r0}
r_{0} \leq 2C_{0} \tau_{n} q_{n}.
\end{equation}

Using  (\ref{r0}) and Lemma \ref{10.1} to bound
$ r_{\ell}$, by (\ref{e3}), we have 
$$|\phi(y)| \leq e^{\left(\delta_{n}/2-L+181 \varepsilon\right) y}.
$$

\textbf{Case 2  of Theorem \ref{local2}.}   
  The proofs of Case 1 and 2 of Theorem \ref{local2} are
completely analogous. We only give a brief proof.
Compared to the Case 1, Case 2 has a lot of simplifications. We don't need to care about the minimum values of (the absolute values of) cosines.

Assume $ \beta_{n}\leq \delta_{n}+200 \varepsilon $, by Proposition \ref{pr1} and \ref{pr2},  bound $c_{n, \ell}$ by 1, we have 
\begin{equation}\label{e1}
|\phi(y)| \leq e^{40 \varepsilon q_{n}} \max \left( e^{-\left(y-2 \ell q_{n}\right) L} r_{\ell}, e^{-\left(2(\ell+1) q_{n}-y\right) L} r_{\ell+1}\right),
\end{equation}
if $2\ell q_{n}+2b_{n}<y<2(\ell+1) q_{n}-2b_{n}$, for some $|\ell| \leq q_{n+1} /\left(6 q_{n}\right)$.
And
for any $\ell \neq 0,|\ell| \leq q_{n+1} /\left(6 q_{n}\right)$, we have
\begin{equation}\label{e2}
r_{\ell} \leq  \frac{e^{-(2 L-70 \varepsilon-\beta_{n}) q_{n}}}{\max (|\ell|, 1)} \max \left(r_{\ell-1}, r_{\ell+1}\right).
\end{equation}
Then similarly to Lemma \ref{10.1}, we have
\begin{Lemma}\label{10.11}
For any $\ell_{0}$ such that $1 \leq\left|\ell_{0}\right| \leq q_{n+1} /\left(6 q_{n}\right)$, we have
$$
r_{\ell_{0}} \leq e^{2\left(\beta_{n}/2-L+54 \varepsilon\right)\left|\ell_{0}\right| q_{n}}
$$
\end{Lemma}
\begin{proof}
 It follows from (\ref{e1}), (\ref{e2}) and by arguments similar to those in Lemma \ref{10.1}. In order to avoid repetition, we omit the details.
\end{proof}

Combing (\ref{r0}), (\ref{e1}) with Lemma \ref{10.11},  thus we have proved  Cases 2 of Theorem \ref{local2}.

\textbf{Case 3 of Theorem \ref{local2}}
 For $k=2n+1,n\in \mathbb{Z}$,
\begin{equation}
    \phi(2n+2)+\phi(2n)=E\phi(2k+1),
\end{equation}
then we have 
\begin{equation}
\label{bbd}
|\phi(2n+1)| \leq C \max(|\phi(2n)|, |\phi(2n+2)|),
\end{equation}
where $C=C(E)$.

Combing  (\ref{bbd}) with Proposition \ref{pr3},  thus we have proved  Cases 3 of Theorem \ref{local2}.

\section{Proofs of some useful Propositions.}\label{0p0}

\subsection{Proofs of Proposition \ref{pr1}.}\label{sig1}
We will first prove not even-resonant $y$ 's can be dominated by even resonances, and then study the relation between adjacent even resonant regions. In the remaining of this paper, by (\ref{bbd}),  we  will only consider the case $k \in 2\mathbb{N}$ without additional statement.
\subsubsection{ $y$ is not even resonance.} 

\begin{Lemma}\label{8.1}
 Assume $2\ell q_{n}+2b_{n} \leq y \leq2(\ell+1) q_{n}-2b_{n}$,  
we have, for $y=2(\ell q_{n}+m_{n})$,
\begin{equation}
|\phi(y)| \leq e^{29 \varepsilon q_{n}} c_{n, \ell} \max \left(e^{-\left(y-2\ell q_{n}\right) L} r_{\ell}, e^{-\left(2(\ell+1) q_{n}-y\right) L} r_{\ell+1}\right).
\end{equation}
For any $y \in  I^{-}$,
\begin{equation}
|\phi(y)| \leq e^{29 \varepsilon q_{n}} \max \left(e^{-\left(y-2\ell q_{n}\right) L} r_{\ell}, c_{n, \ell} e^{-\left(2(\ell+1) q_{n}-y\right) L} r_{\ell+1}\right).
\end{equation}
For any $y \in I^{+}$,
\begin{equation}\label{18.3}
|\phi(y)| \leq e^{29 \varepsilon q_{n}} \max \left(c_{n, \ell} e^{-\left(y-2\ell q_{n}\right) L} r_{\ell}, e^{-\left(2(\ell+1) q_{n}-y\right) L} r_{\ell+1}\right).
\end{equation}
\end{Lemma} 
 We will give the proof of this lemma in the end of the subsection.

 For a not even resonant $y$ and $y \in 2\mathbb{N}$, let $n_{0}$ be the least positive integer so that
$$
4 q_{n-n_{0}} \leq \operatorname{dist}\left(y, 2q_{n} \mathbb{Z}\right) .
$$
Once $n_{0}$ is chosen, we can fix $s$ be the greatest positive integer such that
$$
4 s q_{n-n_{0}} \leq \operatorname{dist}\left(y, 2q_{n} \mathbb{Z}\right).
$$
Clearly, Let
$$
\begin{aligned}
\tilde{I}_{0} &=\left[-\left[s q_{n-n_{0}} / 2\right]-s q_{n-n_{0}}, -\left[s q_{n-n_{0}} / 2\right]\right] \cap \mathbb{Z} ,\\
\tilde{I}_{y} &=\left[y/2-\left[s q_{n-n_{0}} / 2\right]-s q_{n-n_{0}},  y/2-\left[s q_{n-n_{0}} / 2\right]-1\right] \cap \mathbb{Z} .
\end{aligned}
$$
Clearly $\tilde{I}_{0} \cup \tilde{I}_{y}$ contains $2 s q_{n-n_{0}}+1$ distinct numbers. The choice of  $n_{0}$ was first introduced in \cite{liu2015}. It is actually a very useful  technical improvement (simplify   one case appearing in   \cite{ref2}) and now everyone in this field  is  using it as a standard technique.  It should be noted that by the choice of $n_{0}$, we have
$$
2b_{n}<\operatorname{dist}\left(y, 2q_{n} \mathbb{Z}\right)<4 q_{n-n_{0}+1}.
$$
and also
$$
s q_{n-n_{0}}<q_{n-n_{0}+1}.
$$
then we have 
\begin{Lemma}\label{a4.1}
 For a not even-resonant $y$, for $n>N(\varepsilon)$ large enough, we have $\left\{\theta_{\ell}\right\}_{\ell \in \tilde{I}_{0} \cup \tilde{I}_{ y }}$ are $\varepsilon$-uniform.
\end{Lemma}
This is essentially  Lemma 4.1 in \cite{ref81}, we thus omit the proof.
Then we have the following:
\begin{corollary}\label{11.1}
There exists $x_{1} \in \tilde{I}_{0} \cup \tilde{I}_{y}$ such that
$$
\left|\tilde{P}_{4 s q_{n-n_{0}}-1}\left(\theta_{x_{1}}\right)\right| \geq e^{(\tilde{L}-2 \varepsilon)\left(4 s q_{n-n_{0}}-1\right)} .
$$
\end{corollary}
\begin{proof}
It follows from Lemma \ref{up0} and Lemma \ref{a4.1}.
\end{proof}
By a standard argument,  we have the  following:
\begin{Lemma}\label{4.2}
 For $n>N(\varepsilon)$ large enough, there exists $x_{1} \in \tilde{I}_{y}$ so that
$$
\left|\tilde{P}_{4 s q_{n-n_{0}}-1}\left(\theta_{x_{1}}\right)\right| \geq e^{(\tilde{L}-2 \varepsilon)\left(4 s q_{n-n_{0}}-1\right)}.
$$
\end{Lemma}
\begin{proof}
Suppose otherwise, by Corollary \ref{11.1}, we have that for some $x_{1} \in \tilde{I}_{0}$,
\begin{equation}
 \quad\left|\tilde{P}_{4 s q_{n-n_{0}}-1}\left(\theta_{x_{1}}\right)\right| \geq e^{(\tilde{L}-2 \epsilon)\left(4 s q_{n-n_{0}}-1\right)}.
 \end{equation}
Denoting  $x_{2}:=x_{1}+2 s q_{n-n_{0}}-1$ and $I:=\left[2 x_{1}+1,2 x_{2}+1\right]$. By the Green's formula, we have
\begin{equation}\label{con}
\begin{gathered}
|\phi(0)| \leq\left|G_{I}\left(2x_{1}, 0\right)\right| \cdot\left|\phi\left(2x_{1}\right)\right|+\left|G_{I}\left(2x_{2}+1, 0\right)\right| \cdot\left|\phi\left(2x_{2}+2\right)\right| \\
 =\frac{\left|\tilde{P}_{2x_{2}}\left(\theta\right)\right|}{\left|\tilde{P}_{I}\left(\theta_{x_{1}}\right)\right|} \prod_{j=x_{0}}^{0}\left|\cos \left(\pi\left(\theta_{j}\right)\right)\right| \cdot\left|\phi\left(2x_{1}\right)\right|+\frac{\left|\tilde{P}_{-1-2x_{1}}\left(\theta_{x_{1}}\right)\right|}{\left|\tilde{P}_{I}\left(\theta_{x_{1}}\right)\right|} \prod_{j=0}^{x_{2}-1}\left|\cos \left(\pi\left(\theta_{j}\right)\right)\right| \cdot\left|\phi\left(2x_{2}\right)\right|\\ 
 \leq  C_{0} C(\varepsilon) e^{3 \varepsilon|I|}
 \leq  C(\varepsilon) e^{3 \varepsilon|I|}|I| e^{-\frac{\mid I I}{4} L}
< (\varepsilon) e^{-\left(\frac{L}{4}-4 \varepsilon\right)|I|} \rightarrow 0 .
\end{gathered}
\end{equation}
where we used (\ref{eigen}),  Lemma \ref{uper} and Corollary \ref{cos}.

Therefore (\ref{con}) leads to a contradiction with (\ref{eigen1}).

\end{proof}

\begin{proof}[Proof of Lemma \ref{8.1}]

 For $y=2k$, where $k \in \mathbb{Z}$ so that $\operatorname{dist}\left(y, 2q_{n} \mathbb{Z}\right)>2b_{n}$, by Lemma \ref{4.2}, there exists $x_{1} \in I_{k}$ so that
$$
\left|\tilde{P}_{4 s q_{n-n_{0}}-1}\left(\theta_{x_{1}}\right)\right| \geq e^{(\tilde{L}(E)-2 \varepsilon)\left(4 s q_{n-n_{0}}-1\right)} .
$$
Let $x_{2}=x_{1}+2 s q_{n-n_{0}}-1, I(y)=\left[z_{1}, z_{2}\right] \cap \mathbb{Z}=\left[2x_{1}-1, 2x_{2}-1\right] \cap \mathbb{Z} $ and $\partial I(y)=\left\{z_{1}, z_{2}\right\}$.
By Green's function expansion, we have
$$
\phi(y)=\sum_{z \in \partial I(y)} G_{I(y)}(z, y) \phi\left(z^{\prime}\right) .
$$
If $z_{1}=2x_{1}-2>2\ell q_{n}+2b_{n}$ or $z_{2}=2x_{2}<2(\ell+1) q_{n}-2b_{n}$, we could expand $\phi\left(2x_{1}-2\right)$ or $\phi\left(2x_{2}\right)$. We will continue this process until we arrive at a $z$ so that $z \leq 2\ell q_{n}+2b_{n}$ or $z \geq(2\ell+1) q_{n}-2b_{n}$, or the iterating number reaches $t_{0}:=\left[\frac{23}{\tau_{n}}\right]+1$. We obtain, after a series of expansions, the following
$$
\phi(2k)=\sum_{\substack{z_{1},,, z_{t}, z_{t+1} \\ z_{i+1} \in I\left(z_{i}^{\prime}\right)}} G_{I(y)}\left(y, z_{1}\right) G_{I\left(z_{1}^{\prime}\right)}\left(z_{1}^{\prime}, z_{2}\right) \cdots G_{I\left(z_{t}^{\prime}\right)}\left(z_{t}^{\prime}, z_{t+1}\right) \phi\left(z_{t+1}^{\prime}\right),
$$
where $z_{t+1}^{\prime}$ either satisfies

Case 1: $2\ell q_{n} \leq z_{t+1}^{\prime} \leq 2\ell q_{n}+2b_{n}$ and $t<t_{0}$ or,

Case 2: $2(\ell+1) q_{n} \geq z_{t+1}^{\prime} \geq 2(\ell+1) q_{n}-2b_{n}$ and $t<t_{0}$ or,

Case 3: $t=t_{0}$.

 For simplicity, let us denote $y=z_{0}^{\prime}$. 

\noindent If $z_{t+1}^{\prime}$ satisfies Case 1. For each $z_{j}^{\prime}, 0 \leq j \leq t$, denote $\partial I\left(z_{j}^{\prime}\right)=\left\{z_{j+1}, y_{j+1}\right\}$. 
Combing with corollary \ref{cos}, Lemma \ref{4.2} and Lemma \ref{uper}, we have
\begin{equation}
\left|G_{I\left(z_{j}^{\prime}\right)}\left(z_{j}^{\prime}, z_{j+1}\right)\right| \leq C(\varepsilon) e^{-\left|z_{j}^{\prime}-z_{j+1}+1\right|(L-12 \varepsilon)},
\label{11.30}
\end{equation}
furthermore
\begin{equation}
\sum_{\substack{z_{1},,, z_{t}, z_{t+1} \\ z_{i+1} \in I\left(z_{i}^{\prime}\right)}} G_{I(y)}\left(y, z_{1}\right) G_{I\left(z_{1}^{\prime}\right)}\left(z_{1}^{\prime}, z_{2}\right) \cdots G_{I\left(z_{t}^{\prime}\right)}\left(z_{t}^{\prime}, z_{t+1}\right) \phi\left(z_{t+1}^{\prime}\right) \leq (C(\varepsilon))^{t_{0}+1}  e^{-\left(y-2\ell q_{n}-2b_{n}\right)(L-12 \varepsilon)}r_{\ell}.
\label{11.31}
\end{equation}
If $z_{t+1}^{\prime}$ satisfies Case 2 , there must be a $z_{j}^{\prime}$ such that $a q_{n}+m_{n} \in I\left(z_{j}^{\prime}\right)$, we estimate similarly to Case 1, only modifying the estimate of the cosine product, we have
\begin{equation}
\left|G_{I(y)}\left(y, z_{1}\right) G_{I\left(z_{1}^{\prime}\right)}\left(z_{1}^{\prime}, z_{2}\right) \cdots G_{I\left(z_{t}^{\prime}\right)}\left(z_{t}^{\prime}, z_{t+1}\right) \phi\left(z_{t+1}^{\prime}\right)\right| \leq(C(\varepsilon))^{t_{0}+1} e^{\varepsilon q_{n}} e^{-\left((2\ell+1) q_{n}-y\right)(L-12 \varepsilon)} c_{n, \ell} r_{\ell+1} .
\label{11.32}
\end{equation}

If $z_{t+1}^{\prime}$ satisfies Case 3, we bound $|\phi\left(z_{t+1}^{\prime}\right)|$ by
\begin{equation}
\left|\phi\left(z_{t+1}^{\prime}\right)\right| \leq\left\{\begin{array}{l}
\left|\phi\left(x_{0}^{-}\right)\right|, \text {if } z_{t+1}^{\prime} \in I^{-} \\
\left|\phi\left(2(\ell q_{n}+m_{n})\right)\right|, \text { if } z_{t+1}^{\prime}=2\ell q_{n}+2m_{n} \\
\mid \phi\left(x_{0}^{+}\right), \text {if } z_{t+1}^{\prime} \in I^{+}
\end{array}\right..
\label{11.35}
\end{equation}
Using the Green's function estimate (\ref{11.30}),   we have

\begin{equation}
\begin{aligned}
&\left|G_{I(y)}\left(y, z_{1}\right) G_{I\left(z_{1}^{\prime}\right)}\left(z_{1}^{\prime}, z_{2}\right) \cdots G_{I\left(z_{t}^{\prime}\right)}\left(z_{t}^{\prime}, z_{t+1}\right) \phi\left(z_{t+1}^{\prime}\right)\right| \\
&\leq(C(\varepsilon))^{t_{0}} e^{-\frac{1}{4} \tau_{n}q_{n} t_{0}(L-12 \varepsilon)} \max \left\{\left|\phi\left(x_{0}^{-}\right)\right|,\left|\phi\left(2\ell q_{n}+2m_{n}\right)\right|, c_{n, \ell}\left|\phi\left(x_{0}^{+}\right)\right|\right\} \\
&\leq e^{-6 q_{n}(L-12 \varepsilon)} \max \left\{\left|\phi\left(x_{0}^{-}\right)\right|,\left|\phi\left(2\ell q_{n}+2m_{n}\right)\right|, c_{n, \ell}\left|\phi\left(x_{0}^{+}\right)\right|\right\}.
\end{aligned}
\label{11.36}
\end{equation}
Taking into account all the three cases (\ref{11.31}), (\ref{11.32}) and (\ref{11.36}), we have proved that for even point $y \in I^{-}$,
\begin{equation}
\begin{gathered}
|\phi(y)| \leq(C(\varepsilon))^{t_{0}} \max \left(e^{\varepsilon q_{n}} e^{-\left(y-2\ell q_{n}\right)(L-12 \varepsilon)} r_{\ell}, e^{\varepsilon q_{n}} e^{-\left(2(\ell+1) q_{n}-y\right)(L-12 \varepsilon)} c_{n, \ell} r_{\ell+1},\right. \\
\left.e^{-3 q_{n}(L-12 \varepsilon)} \max \left(\left|\phi\left(x_{0}^{-}\right)\right|,\left|\phi\left(2\ell q_{n}+2m_{n}\right)\right|, c_{n, \ell}\left|\phi\left(x_{0}^{+}\right)\right|\right)\right).
\end{gathered}
\label{11.37}
\end{equation}
Letting $y= x_{0}^{-}$, we have $\left|\phi\left(x_{0}^{-}\right)\right| \leq(C(\varepsilon))^{t_{0}} \max \left(r_{\ell}, r_{\ell+1}\right) . $

Similarly, one can show that for $y \in I_{+}$,
\begin{equation}
 \label{11.38}   
\begin{array}{r}
|\phi(y)| \leq(C(\varepsilon))^{t_{0}} \max \left(e^{\varepsilon q_{n}} e^{-\left(y-2\ell q_{n}\right)(L-12 \varepsilon)} c_{n, \ell} r_{\ell}, e^{\varepsilon q_{n}} e^{-\left(2(\ell+1) q_{n}-y\right)(L-12 \varepsilon)} r_{\ell+1},\right. \\
\left.e^{-3 q_{n}(L-12 \varepsilon)} \max \left(c_{n, \ell}\left|\phi\left(x_{0}^{-}\right)\right|,\left|\phi\left(2\ell q_{n}+2m_{n}\right)\right|,\left|\phi\left(x_{0}^{+}\right)\right|\right)\right).
\end{array}
\end{equation}
and
\begin{equation}
\begin{array}{r}
\left|\phi\left(2\ell q_{n}+2m_{n}\right)\right| \leq(C(\varepsilon))^{t_{0}} c_{n, \ell} \max \left(e^{\varepsilon q_{n}} e^{-\left(y-2\ell q_{n}\right)(L-12 \varepsilon)} r_{\ell}, e^{\varepsilon q_{n}} e^{-\left(2(\ell+1) q_{n}-y\right)(L-12 \varepsilon)} r_{\ell+1}\right. \\
\left.e^{-3 q_{n}(L-12 \varepsilon)} \max \left(\phi\left(x_{0}^{-}\right)|,| \phi\left(2\ell q_{n}+2m_{n}\right)|,| \phi\left(x_{0}^{+}\right) \mid\right)\right).
\end{array}
\label{11.39}
\end{equation}
Letting $y= x_{0}^{+}$ in (\ref{11.38}),  together with (\ref{11.39}), we have
$$
\max \left(\left|\phi\left(x_{0}^{-}\right)\right|,\left|\phi\left(x_{0}^{+}\right)\right|,\left|\phi\left(a q_{n}+m_{n}\right)\right|\right) \leq(C(\varepsilon))^{t_{0}} \max \left(r_{\ell}, r_{\ell+1}\right) .
$$
Plugging this back into (\ref{11.37}), (\ref{11.38}), (\ref{11.39}), we obtain the claimed result for all even points.
Combing these with (\ref{bbd}), we obtain the claimed result for all $y \in \mathbb{N}$ .
\end{proof} 

\subsubsection{$y$ is even-resonance.}
\begin{Lemma}\label{8.2}
For any $\ell \neq 0,|\ell| \leq q_{n+1} /\left(6 q_{n}\right)$,
$$
r_{\ell} \leq  \frac{e^{-(2 L -55 \varepsilon) q_{n}}}{\max (|\ell|, 1)} \max \left(r_{\ell-1}, r_{\ell+1}\right) \times \begin{cases}\max \left(|\ell|, e^{\delta_{n} q_{n}}\right), & \text { if } \beta_{n} \geq \delta_{n}+200 \varepsilon \\ e^{\beta_{n} q_{n}}, & \text { if } \beta_{n}<\delta_{n}+200 \varepsilon\end{cases}.
$$
\end{Lemma}
\begin{proof}
For $\ell \in \mathbb{Z}$, let $I_{\ell}$ be defined below
$$
I_{\ell}:=\left[(\ell-1) q_{n}-\left\lfloor q_{n} / 2\right\rfloor, \ell q_{n}-\left\lfloor q_{n} / 2\right\rfloor-1\right] \cap \mathbb{Z} .
$$ for $\ell>0$
and 
$$
I_{0}:=\left[- q_{n}-\left\lfloor q_{n} / 2\right\rfloor, q_{n}-\left\lfloor q_{n} / 2\right\rfloor\right] \cap \mathbb{Z} .
$$
\begin{Lemma}[Lemma 4.3, \cite{ref81}]\label{resonunif}
For $\ell$ such that $0<|\ell| \leq 2 q_{n+1} /\left(3 q_{n}\right),\left\{\theta_{j}\right\}_{j \in I_{0} \cup I_{\ell}}$ are $\frac{\ln \left(q_{n+1} /|\ell|\right)}{2 q_{n}-1}+\epsilon$-uniform.
\end{Lemma}
 Combing this with Lemma \ref{up}, we have the following:
\begin{corollary}\label{4.4}
For $\ell$ such that $0<|\ell| \leq 2 q_{n+1} /\left(3 q_{n}\right)$, there exists $x_{1} \in I_{0} \cup I_{\ell}$ such that
$$
\left|\tilde{P}_{4 q_{n}-1}\left(\theta_{x_{1}}\right)\right| \geq \frac{|\ell|}{q_{n+1}} e^{(\tilde{L}-2 \varepsilon)\left(4 q_{n}-1\right)}.
$$
\end{corollary}

More precisely, 
\begin{Lemma}\label{8.6}
 For any $\ell \neq 0,|\ell| \leq q_{n+1} /\left(6 q_{n}\right)$, there exists $x_{1} \in  I_{\ell}$ such that
$$
\left|\tilde{P}_{4 q_{n}-1}\left(\theta_{x_{1}}\right)\right| \geq \max (|\ell|, 1) e^{-\beta_{n} q_{n}} e^{(\tilde{L}-2 \varepsilon)\left(4 q_{n}-1\right)}.
$$
\end{Lemma}
\begin{proof}
Similar to  the argument in the proof of Lemma \ref{4.2}, we have that for any $x_{1} \in {I}_{0}$ so that
$$
\left|\tilde{P}_{4 s q_{n}-1}\left(\theta_{x_{1}}\right)\right| \leq \frac{|\ell|}{q_{n+1}} e^{(\tilde{L}-2 \varepsilon)\left(4 s q_{n}-1\right)}.
$$
Therefore Lemma \ref{8.6} holds.

\end{proof}
\begin{Lemma}
 Assume that there exists $x_{1} \in I_{\ell}$ such that
\begin{equation}\label{18.4}
\left|\tilde{P}_{4 q_{n}-1}\left(\theta_{x_{1}}\right)\right| \geq \max (|\ell|, 1) e^{-\beta_{n} q_{n}} e^{(\tilde{L}-2 \varepsilon)\left(4 q_{n}-1\right)} .
\end{equation}
Then we have
$$
r_{\ell} \leq \frac{e^{-(2L-\beta_{n}-55 \varepsilon) q_{n}}}{\max (|\ell|, 1)}  \max \left(c_{n, \ell-1} r_{\ell-1}, c_{n, \ell} r_{\ell+1}\right).
$$
\label{8.3}
\end{Lemma}
This is a variant of Lemma 8.3 in \cite{ref81}. If $\beta_{n} \geq \delta_{n}+200 \varepsilon$, bound the $c_{n, j}$ 's by (\ref{5.38}). Otherwise trivially bound the $c_{n, j}$ 's by 1, then Lemma \ref{8.2} follows from combining  Lemma (\ref{8.3}) with Lemma \ref{8.6}.
\end{proof}
 Proposition \ref{pr1} follows directly by Lemma \ref{8.1} and  Lemma \ref{8.2}.

\subsection{Proofs of Proposition \ref{pr2}.}\label{sig2}

\subsubsection{$y$ is not even-resonance. }
\begin{Lemma}\label{9.1}
If $2\ell q_{n}+2b_{n}<y<2(\ell+1) q_{n}-2b_{n}$, for some $|\ell| \leq q_{n+1} /\left(6 q_{n}\right)$. Then
$$
|\phi(y)| \leq e^{40 \varepsilon q_{n}} \max \left(e^{-\left(y-2\ell q_{n}\right) L} r_{\ell}^{+}, e^{-\left(2(\ell+1) q_{n}-y\right) L} r_{\ell+1}^{-}\right) .
$$
\end{Lemma} 
\begin{proof}
 The proof of this lemma is almost identical to that of Lemma \ref{8.1}. We only give a brief proof. By Green's function expansion, we have
$$
\phi(y)=\sum_{z \in \partial I(y)} G_{I(y)}(z, y) \phi\left(z^{\prime}\right) .
$$
If $2x_{1}-2>2\ell q_{n}+2b_{n}$ or $2x_{2}<2(\ell+1) q_{n}-2b_{n}$, we continue to expand $\phi\left(2x_{1}-2\right)$ or $\phi\left(2x_{2}\right)$. 
We repeat this process until we arrive at a $z$ so that $z \leq 2\ell q_{n}+2b_{n}$ or $z \geq 2(\ell+1) q_{n}-2b_{n}$, or the iterating number reaches $t_{0}:=\left[24 / \tau_{n}\right]+1$. We obtain, after a series of expansions, the following
\begin{equation}
\phi(y)=\sum_{s ; z_{i+1} \in I\left(z_{i}^{\prime}\right)} G_{I(y)}\left(y, z_{1}\right) G_{I\left(z_{1}^{\prime}\right)}\left(z_{1}^{\prime}, z_{2}\right) \cdots G_{I\left(z_{t}^{\prime}\right)}\left(z_{t}^{\prime}, z_{t+1}\right) \phi\left(z_{t+1}^{\prime}\right),
\end{equation}
where $z_{t+1}^{\prime}$ either satisfies

Case 1: $z_{t+1}^{\prime} \in R_{\ell}^{+} \cup\left\{2\ell q_{n}+2m_{n}\right\}$ and $t<t_{0}$ or,

Case 2: $z_{t+1}^{\prime} \in R_{\ell}^{-}$and $t<t_{0}$ or,

Case 3: $z_{t+1}^{\prime} \in R_{\ell+1}^{-}$and $t<t_{0}$ or,

Case 4 $: t=t_{0}$.

Therefore, we have
 \begin{equation}\label{11.46}
|\phi(y)| \leq(C(\varepsilon))^{t_{0}} e^{18 \varepsilon q_{n}} \max \left(e^{-\left(y-2\ell q_{n}\right) L} r_{\ell}^{+}, c_{n, \ell} e^{-\left(y-2\ell q_{n}\right) L} r_{\ell}^{-}, e^{-\left(2(\ell+1) q_{n}-y\right) L} r_{\ell+1}^{-}\right).
\end{equation} 
Then, we will use the following lemmas to study the relation of $r_{\ell}^{-}$ and $r_{\ell}^{+}$.
\begin{Lemma}( Corollary 5.8, \cite{ref81})\label{5.8}
 Let $I=\left[2\ell_{1}, 2\ell_{2}\right] \subset \mathbb{Z}$ be such that $2\ell_{1} \in\left[2(j-1) q_{n}+2m_{n}-2, 2j q_{n}+2m_{n}-2\right]$ and $2\ell_{2} \in\left[2j q_{n}+2m_{n}+2,2(j+1) q_{n}+2m_{n}-2\right]$, for some $j \in \mathbb{Z},|j|<q_{n+1} /\left(6 q_{n}\right)$. For $n>N(\varepsilon)$ large enough we have
$$
\left\|A_{|I/2|}\left(\theta_{\ell_{1}}\right)\right\| \leq e^{7 \varepsilon q_{n}} \frac{1}{c_{n, j}} e^{L|I|}.
$$
\end{Lemma}
Thus we have
$$
r_{\ell}^{-} \leq e^{18 \varepsilon q_{n}} \frac{1}{c_{n, \ell}} r_{\ell}^{+}.
$$
Hence (\ref{11.46}) yields
\begin{equation}
\quad|\phi(y)| \leq e^{40 \varepsilon q_{n}} \max \left(e^{-\left(y-2\ell q_{n}\right) L} r_{\ell}^{+}, e^{-\left(2(\ell+1) q_{n}-y\right) L} r_{\ell+1}^{-}\right).
\end{equation}
This proves the claimed result.
\end{proof}

\subsubsection{$y$ is even-resonance.}
Assume without loss of generality that $0<m_{n} \leq b_{n}$.
The main lemma of this section is the following.
\begin{Lemma}\label{9.2}
  For any $\ell \neq 0$ such that $|\ell|<q_{n+1} /\left(6 q_{n}\right)$, we have
\begin{equation}
r_{\ell} \leq \frac{e^{-(2 L- 70 \varepsilon) q_{n}}}{\max (|\ell|, 1)} \max \left(r_{\ell-1}, r_{\ell+1}\right) \times\left\{\begin{array}{l}
\max \left(|\ell|, e^{\delta_{n} q_{n}}\right), \text { if } \beta_{n} \geq \delta_{n}+200 \varepsilon \\
e^{\beta_{n} q_{n}}, \text { if } \beta_{n}<\delta_{n}+200 \varepsilon
\end{array}\right..
\end{equation}
\end{Lemma}

\begin{proof}
This argument is very similar to that of Lemma \ref{8.2}.
Firstly, we need the following Lemma:
\begin{Lemma}\label{9.3}
 Assume that there exists $x_{1} \in I_{\ell}$, for some $|\ell|<q_{n+1} /\left(6 q_{n}\right)$, such that
\begin{equation}\label{19.2}
\left|\tilde{P}_{4 q_{n}-1}\left(\theta_{x_{1}}\right)\right| \geq \max (|\ell|, 1) e^{-\beta_{n} q_{n}} e^{(\tilde{L}-2 \varepsilon)\left(4 q_{n}-1\right)} .
\end{equation}
We  have
$$
r_{\ell}^{-} \leq  \frac{e^{-(2 L-\beta_{n}-69 \varepsilon  )q_{n}}}{\max (|\ell|, 1)}  \max \left(c_{n, \ell-1} r_{\ell-1}^{-}, c_{n, \ell-1} r_{\ell-1}^{+}, \gamma r_{\ell-1}^{+}, c_{n, \ell} r_{\ell}^{+}, c_{n, \ell} r_{\ell+1}^{-}, c_{n, \ell} c_{n, \ell+1} r_{\ell+1}^{+}\right),
$$
and
$$
r_{\ell}^{+} \leq  \frac{e^{-(2 L-\beta_{n}-69 \varepsilon  )q_{n}}}{\max (|\ell|, 1)} \max \left(c_{n, \ell} c_{n, \ell-1} r_{\ell-1}^{-}, c_{n, \ell} r_{\ell-1}^{+}, c_{n, \ell} r_{\ell}^{-}, \gamma r_{\ell+1}^{-}, c_{n, \ell+1} r_{\ell+1}^{-}, c_{n, \ell+1} r_{\ell+1}^{+}\right).
$$
where
$$
\gamma:=\left\{\begin{array}{l}
\max \left(e^{\delta_{n} q_{n}},|\ell|, 1\right) e^{-\beta_{n} q_{n}}, \text { if } \beta_{n} \geq \delta_{n}+200 \varepsilon \\
1, \text { otherwise }
\end{array}\right..
$$
\end{Lemma}
This is a variant of Lemma 9.3 in \cite{ref81}. 
If $\beta_{n} \geq \delta_{n}+200 \varepsilon$, bound the $c_{n, j}$ 's by (\ref{5.38}). Otherwise trivially bound the $c_{n, j}$ 's by 1, then combining Corollary Lemma \ref{8.6} with Lemma \ref{9.3},
 for any $\ell \neq 0$ such that $|\ell| \leq q_{n+1} /\left(6 q_{n}\right)$,  the following
hold
$$
r_{\ell} \leq  \frac{e^{-(2 L-70 \varepsilon)q_{n}}}{\max (|\ell|, 1)} \max \left(r_{\ell-1}, r_{\ell}, r_{\ell+1}\right) \times\left\{\begin{array}{l}
\max \left(|\ell|, e^{\delta_{n} q_{n}}\right), \text { if } \beta_{n} \geq \delta_{n}+200 \varepsilon \\
e^{\beta_{n} q_{n}}, \text { if } \beta_{n}<\delta_{n}+200 \varepsilon
\end{array}\right..
$$
It should be noted that 
$$
 \frac{e^{-(2L-70 \varepsilon )q_{n}}}{\max (|\ell|, 1)}  \max \left(|\ell|, e^{\delta_{n} q_{n}}, 1\right)  \leq e^{-\left(2L-\delta_{n}-70 \varepsilon\right) q_{n}}<1,
$$
so the $r_{\ell}$ terms on the right-hand-side of the equation above can be dropped. This proves Lemma \ref{9.2}.
\end{proof}

Proposition \ref{pr2} is a consequence of  Lemma \ref{9.1} and  Lemma \ref{9.2}.

\section*{Acknowledgements} 
\quad   The authors would like to thank Q.Zhou for giving this problem and valuable suggestions. The authors would also like to thank J.You for useful comments. This work was   supported by  Nankai Zhide Foundation.

\begin{appendices}

\section{Proofs of Lemma \ref{av4}.}

 We have
$$
\tilde{P}_{2k}(2 \theta)=\operatorname{det}\left[\begin{array}{ccccc}
t_{1} & c_{1} & & & \\
c_{2} & t_{2} & c_{2} & & \\
& c_{3} & \cdots & & \\
& & & \cdots & c_{2k-1} \\
& & & c_{2k} & t_{2k}
\end{array}\right]_{2k \times 2k}
$$
where  $t_{2j} \triangleq E \cos 2 \pi\left(\theta+j \alpha\right)-\lambda \sin 2 \pi\left(\theta+j\alpha\right)$, $t_{2j+1} \triangleq E $, $c_{2j} \triangleq-\cos 2 \pi\left(\theta+j \alpha\right)$ and $c_{2j+1} \triangleq -1$. 
Then
$$
\left\{\begin{array}{l}
\tilde{t}_{2j}(z) \triangleq e^{2\pi i j \alpha} z \cdot t_{2j}(z)=\quad \frac{E+i \lambda}{2} e^{4 i \pi j \alpha} z^{2}+\frac{E-i \lambda}{2}, \\
\tilde{c}_{2j}(z) \triangleq e^{2\pi i j \alpha} z \cdot c_{2j}(z)=\quad-\frac{1}{2} e^{4 i \pi j \alpha} z^{2}-\frac{1}{2} .
\end{array}\right.
$$
and
$$
\left\{\begin{array}{l}
\tilde{t}_{2j+1}(z) \triangleq t_{2j+1}(z), \\
\tilde{c}_{2j+1}(z) \triangleq c_{2j+1}(z).
\end{array}\right.
$$
Since $|z|=1$, we have
\begin{equation}
\left|\tilde{P}_{2k}(2 \theta)\right|=\left|f_{k}(z)\right|=\left|\operatorname{det}\left[\begin{array}{ccccc}
\tilde{t}_{1}(z) & \tilde{c}_{1}(z) & & & \\
\tilde{c}_{2}(z) & \tilde{t}_{2}(z) & \tilde{c}_{2}(z) & & \\
& \tilde{c}_{3}(z) & \cdots & & \\
& & & \ldots & \tilde{c}_{k-2}(z) \\
& & & \tilde{c}_{2k}(z) & \tilde{t}_{2k}(z)
\end{array}\right]_{2k \times 2k}\right|.
\end{equation}
Clearly, $\ln \left|f_{k}(z)\right|$ is a subharmonic function, therefore
$$
\frac{1}{k} \int_{\mathbb{T}} \ln \left|\tilde{P}_{k}(2 \theta)\right| \mathrm{d} \theta=\frac{1}{k} \int_{\mathbb{T}} \ln \left|f\left(e^{2 \pi i \theta}\right)\right| \mathrm{d} \theta \geq \frac{1}{k} \ln \left|f_{k}(0)\right|.
$$

\begin{equation}
\begin{aligned}
&f_{2k}(0)=\operatorname{det}\left[\begin{array}{ccccccc}
E  & -1  & & & &   \\
-1 / 2 & (E-i \lambda) / 2 & -1 / 2  & & &\\
&-1 & E & -1 & & & \\
& &-1 / 2 & (E-i \lambda) / 2 & -1 / 2  & &\\
& &  & -1/2  & \cdots & & \\
& & &  &  \cdots & -1 & \\
& & & & &  -1 / 2 & (E-i \lambda) / 2
\end{array}\right]_{2k \times 2k}\\
&=\frac{1}{(2)^{k}} \operatorname{det}\left[\begin{array}{ccccccc}
-E & 1 & & & & \\
1 & i \lambda-E & 1 & & &\\
&1 & -E & 1 & & & \\
& & 1 & i \lambda-E & 1  & &\\
& &  & 1  & \cdots & & \\
& & &  &  \cdots & 1 & \\
& & & & &  1& i \lambda -E
\end{array}\right]_{2k \times 2k} \triangleq  \frac{1}{(2)^{k}} a_{2k} .
\end{aligned}
\end{equation}
Similarly, we can denote $f_{2k+1}(0)= \frac{-1}{(2)^{k}} b_{2k+1} $.
Obviously $b_{2k+1}=(i \lambda-E) a_{2k}-a_{2k-1}$ and $a_{2k}=-E b_{2k-1}-b_{2k-2}$.
Thus
$$
\left|a_{2k}\right| \sim \left|b_{2k+1}\right| \sim C\left|x_{2}\right|^{k} \text { as } k \rightarrow \infty,
$$
where $\left|x_{1}\right|<1<\left|x_{2}\right|$ are solutions of the characteristic equation
$$
x^{2}-(E^2-i \lambda E-2) x+1=0 .
$$
Therefore   we have
$$
\lim _{k \rightarrow \infty} \frac{1}{k} \int_{0}^{1} \ln \left|\tilde{P}_{k}(\theta)\right| \mathrm{d} \theta \geq \ln \left|x_{2}\right|-\frac{\ln 2}{2} .
$$
Then the result  follows from  Lemma \ref{le2}.
\

\end{appendices}

\bibliographystyle{siam}
\bibliography{sample}

\begin{thebibliography}{10}

\bibitem{AA80}
{\sc S.~Aubry and G.~Andr{\'e}}, {\em Analyticity breaking and anderson
  localization in incommensurate lattices}, Ann. Israel Phys. Soc, 3 (1980),
  p.~18.

\bibitem{Aab}
{\sc A.~Avila}, {\em Absolutely continuous spectrum for the almost mathieu
  operator with subcritical coupling}, Preprint,  (2008).

\bibitem{ref1}
{\sc A.~Avila}, {\em Global theory of one-frequency schrödinger operators},
  Acta Mathematica, 215 (2015), pp.~1--54.

\bibitem{AD}
{\sc A.~Avila and D.~Damanik}, {\em Absolute continuity of the integrated
  density of states for the almost mathieu operator with non-critical
  coupling}, Inventiones mathematicae, 172 (2008), pp.~439--453.

\bibitem{ref2}
{\sc A.~Avila and S.~Jitomirskaya}, {\em The {Ten} {Martini} problem}, Ann.
  Math. (2), 170 (2009), pp.~303--342.

\bibitem{ref23}
{\sc A.~Avila, S.~Jitomirskaya, and Q.~Zhou}, {\em Second phase transition
  line}, Mathematische Annalen, 370 (2018).

\bibitem{ref3}
{\sc A.~Avila and S.~Y. Jitomirskaya}, {\em Almost localization and almost
  reducibility}, arXiv: Dynamical Systems,  (2008).

\bibitem{ayz}
{\sc A.~Avila, J.~You, and Q.~Zhou}, {\em Sharp phase transitions for the
  almost mathieu operator}, Duke Mathematical Journal, 166 (2017),
  pp.~2697--2718.

\bibitem{bellissard1983}
{\sc J.~Bellissard, R.~Lima, and E.~Scoppola}, {\em Localization
  inv-dimensional incommensurate structures}, Communications in Mathematical
  Physics, 88 (1983), pp.~465--477.

\bibitem{berry1984}
{\sc M.~Berry}, {\em Incommensurability in an exactly-soluble quantal and
  classical model for a kicked rotator}, Physica D: Nonlinear Phenomena, 10
  (1984), pp.~369--378.

\bibitem{ref26}
{\sc J.~Bourgain and S.~Y. Jitomirskaya}, {\em Continuity of the lyapunov
  exponent for quasiperiodic operators with analytic potential}, Journal of
  Statistical Physics, 108 (2001), pp.~1203--1218.

\bibitem{ref21}
{\sc D.~{Damanik}, J.~{Fillman}, and P.~{Gohlke}}, {\em {Spectral
  Characteristics of Schr{\"o}dinger Operators Generated by Product Systems}},
  arXiv e-prints,  (2022), p.~arXiv:2203.11739.

\bibitem{figotin1984}
{\sc A.~Figotin and L.~Pastur}, {\em An exactly solvable model of a
  multidimensional incommensurate structure}, Communications in mathematical
  physics, 95 (1984), pp.~401--425.

\bibitem{fishman10}
{\sc S.~Fishman}, {\em Anderson localization and quantum chaos maps},
  Scholarpedia, 5 (2010), p.~9816.

\bibitem{ref6}
{\sc S.~Fishman, D.~R. Grempel, and R.~E. Prange}, {\em Localization in a
  d-dimensional incommensurate structure}, Physical Review B, 29 (1984),
  pp.~4272--4276.

\bibitem{ref10}
{\sc A.~Furman}, {\em On the multiplicative ergodic theorem for uniquely
  ergodic systems}, Annales De L Institut Henri Poincare-probabilites Et
  Statistiques, 33 (1997), pp.~797--815.

\bibitem{ganeshan2014}
{\sc S.~Ganeshan, K.~Kechedzhi, and S.~D. Sarma}, {\em Critical integer quantum
  hall topology and the integrable maryland model as a topological quantum
  critical point}, Physical Review B, 90 (2014), p.~041405.

\bibitem{ref27}
{\sc R.~Han}, {\em Shnol’s theorem and the spectrum of long range operators},
  Proceedings of the American Mathematical Society,  (2019).

\bibitem{ref81}
{\sc R.~Han, S.~Jitomirskaya, and F.~Yang}, {\em Anti-resonances and sharp
  analysis of maryland localization for all parameters}, arXiv preprint
  arXiv:2205.04021,  (2022).

\bibitem{Ha}
{\sc P.~G. Harper}, {\em Single band motion of conduction electrons in a
  uniform magnetic field}, Proceedings of the Physical Society. Section A, 68
  (1955), p.~874.

\bibitem{H}
{\sc M.~R. Herman}, {\em Une m{\'e}thode pour minorer les exposants de
  lyapounov et quelques exemples montrant le caractere local d’un
  th{\'e}oreme d’arnold et de moser sur le tore de dimension 2}, Commentarii
  Mathematici Helvetici, 58 (1983), pp.~453--502.

\bibitem{ref11}
{\sc S.~Jitomirskaya}, {\em Metal-insulator transition for the almost mathieu
  operator}, Annals of Mathematics, 150 (1999), pp.~1159--1175.

\bibitem{ref8}
{\sc S.~Jitomirskaya and W.~Liu}, {\em Arithmetic spectral transitions for the
  maryland model}, Communications on Pure and Applied Mathematics, 70 (2016).

\bibitem{ref84}
\leavevmode\vrule height 2pt depth -1.6pt width 23pt, {\em Universal
  hierarchical structure of quasiperiodic eigenfunctions}, Annals of
  Mathematics, 187 (2018), pp.~721--776.

\bibitem{JLiu1}
\leavevmode\vrule height 2pt depth -1.6pt width 23pt, {\em Universal
  reflective-hierarchical structure of quasiperiodic eigenfunctions and sharp
  spectral transition in phase}, arXiv preprint arXiv:1802.00781,  (2018).

\bibitem{ref24}
{\sc S.~Jitomirskaya and C.~A. Marx}, {\em Analytic quasi-perodic cocycles with
  singularities and the lyapunov exponent of extended harper's model}, 2010.

\bibitem{ref25}
{\sc S.~Jitomirskaya and C.~A. Marx}, {\em Erratum to: Analytic quasi-perodic
  cocycles with singularities and the lyapunov exponent of extended harper’s
  model}, Communications in Mathematical Physics, 317 (2013), pp.~269--271.

\bibitem{2017s}
{\sc S.~Jitomirskaya and F.~Yang}, {\em Singular continuous spectrum for
  singular potentials}, Communications in mathematical physics, 351 (2017),
  pp.~1127--1135.

\bibitem{ref13}
{\sc S.~Jitomirskaya and F.~Yang}, {\em Pure point spectrum for the maryland
  model: a constructive proof}, Ergodic Theory and Dynamical Systems, 41
  (2021), pp.~283 -- 294.

\bibitem{kachkovskiy2019}
{\sc I.~Kachkovskiy}, {\em Localization for quasiperiodic operators with
  unbounded monotone potentials}, Journal of Functional Analysis, 277 (2019),
  pp.~3467--3490.

\bibitem{kachkovskiy2021}
{\sc I.~Kachkovskiy, S.~Krymski, L.~Parnovski, and R.~Shterenberg}, {\em
  Perturbative diagonalization for maryland-type quasiperiodic operators with
  flat pieces}, Journal of Mathematical Physics, 62 (2021), p.~063509.

\bibitem{L93}
{\sc Y.~Last}, {\em A relation between ac spectrum of ergodic jacobi matrices
  and the spectra of periodic approximants}, Communications in mathematical
  physics, 151 (1993), pp.~183--192.

\bibitem{liu2015}
{\sc W.~Liu and X.~Yuan}, {\em Anderson localization for the completely
  resonant phases}, Journal of Functional Analysis, 268 (2015), pp.~732--747.

\bibitem{longhi2021}
{\sc S.~Longhi}, {\em Maryland model in optical waveguide lattices}, Optics
  Letters, 46 (2021), pp.~637--640.

\bibitem{Pe}
{\sc R.~Peierls}, {\em Zur theorie des diamagnetismus von leitungselektronen},
  Zeitschrift f{\"u}r Physik, 80 (1933), pp.~763--791.

\bibitem{R}
{\sc A.~Rauh}, {\em Degeneracy of landau levels in crystals}, physica status
  solidi (b), 65 (1974), pp.~K131--K135.

\bibitem{simon85}
{\sc B.~Simon}, {\em Almost periodic schr{\"o}dinger operators iv. the maryland
  model}, Annals of Physics, 159 (1985), pp.~157--183.

\bibitem{simon1989}
{\sc B.~Simon and T.~Spencer}, {\em Trace class perturbations and the absence
  of absolutely continuous spectra}, Communications in mathematical physics,
  125 (1989), pp.~113--125.

\bibitem{ref9}
{\sc P.~Walters}, {\em Ergodic theory and differentiable dynamics by ricardo
  ma{\~n}{\'e}: Translated from the portuguese by silvio levy. ergebnisse de
  mathematik und ihrer grenzgebiete, 3 folge-band 8. springer-verlag 1987.},
  Ergodic Theory and Dynamical Systems, 9 (1989), pp.~399 -- 401.

\bibitem{ref22}
{\sc Y.~Wang, X.~Xia, J.~You, Z.~Zheng, and Q.~Zhou}, {\em Exact mobility edges
  for 1d quasiperiodic models}, arXiv preprint arXiv:2110.00962,  (2021).

\bibitem{wang2020}
{\sc Y.~Wang, X.~Xia, L.~Zhang, H.~Yao, S.~Chen, J.~You, Q.~Zhou, and X.-J.
  Liu}, {\em One-dimensional quasiperiodic mosaic lattice with exact mobility
  edges}, Physical Review Letters, 125 (2020), p.~196604.

\end{thebibliography}

\end{document}